\documentclass[twosided,final]{article}
\usepackage[utf8]{inputenc}
\usepackage{graphicx}
\usepackage{fullpage}
\usepackage[english]{babel} 
\usepackage{amssymb,mathtools}
\usepackage{xspace}
\usepackage{booktabs}
\usepackage{boxedminipage}
\usepackage{tabularx}
\usepackage{calc}
\usepackage{hyperref}
\usepackage{csquotes}
\usepackage{color}
\usepackage{microtype}

\usepackage{tikz}
\usetikzlibrary{arrows,positioning,backgrounds,fit,shapes,calc,scopes}

\usepackage[textwidth=20mm,obeyFinal,colorinlistoftodos]{todonotes}

\def\ve#1{\mathchoice{\mbox{\boldmath$\displaystyle\bf#1$}}
    {\mbox{\boldmath$\textstyle\bf#1$}}
    {\mbox{\boldmath$\scriptstyle\bf#1$}}
    {\mbox{\boldmath$\scriptscriptstyle\bf#1$}}}

\newcommand\veb{{\ve b}}
\newcommand\vecc{{\ve c}}
\newcommand\ved{{\ve d}}
\newcommand\vece{{\ve e}}
\newcommand\vef{{\ve f}}
\newcommand\veg{{\ve g}}
\newcommand\veh{{\ve h}}
\newcommand\vel{{\ve l}}

\newcommand\ves{{\ve s}}

\newcommand\veu{{\ve u}}
\newcommand\vev{{\ve v}}
\newcommand\vew{{\ve w}}
\newcommand\vex{{\ve x}}
\newcommand\vey{{\ve y}}
\newcommand\vez{{\ve z}}
\newcommand\vealpha{{\ve \alpha}}

\newcommand\vesigma{{\ve \sigma}}
\def \A {E^{(n)}}
\def\R{\mathbb{R}}
\def\Z{\mathbb{Z}}
\def\N{\mathbb{N}}
\def\RR{\mathbb{R}}
\def \C {{\cal C}}
\def \F {{\cal F}}
\def \G {{\cal G}}

\def \L {{\cal L}}
\def \R {{\cal R}}
\def \T {{\cal T}}

\def \l {\langle}
\def \r {\rangle}
\newcommand{\pref}{\ensuremath{\succ}}
\newcommand{\Oh}{O}

\newcommand{\cGEn}{\G\big(\A\big)}

\DeclareMathOperator{\rank}{{\rm rank}}

\newcommand{\NPh}{\hbox{{\sf NP}-hard}\xspace}
\newcommand{\NPc}{\hbox{{\sf NP}-complete}\xspace}

\newcommand{\transpose}{\intercal}
\newcommand{\hy}{\hbox{-}\nobreak\hskip0pt}
\newcommand{\prob}[3]{
    \begin{center}
        \begin{tabularx}{\textwidth}{lX}
            \multicolumn{2}{l}{#1}\\
            {\bf Input:}&{#2}\\
            {\bf Find:}&{#3}
        \end{tabularx}
    \end{center}
}

\usepackage{amsthm}

\theoremstyle{plain}
\newtheorem{theorem}{Theorem}
\newtheorem{corollary}[theorem]{Corollary}

\newtheorem{lemma}[theorem]{Lemma}
\newtheorem{proposition}[theorem]{Proposition}

\theoremstyle{definition}
\newtheorem{definition}[theorem]{Definition}

\newcommand{\OhOp}[1]{O\mathopen{}\mathclose\bgroup\left( #1 \aftergroup\egroup\right)}

\makeatletter
\newtheorem*{rep@theorem}{\rep@title}
\newcommand{\newreptheorem}[2]{%
\newenvironment{rep#1}[1]{%
 \def\rep@title{#2 \ref{##1}}%
 \begin{rep@theorem}}%
 {\end{rep@theorem}}}
\makeatother

\newreptheorem{theorem}{Theorem}

\newcommand{\runtimeclass}[1]{\ensuremath{\mathsf{#1}}}
\newcommand{\FPT}{\runtimeclass{FPT}\xspace}
\newcommand{\NP}{\runtimeclass{NP}\xspace}

\begin{document}

\title{Combinatorial $n$-fold Integer Programming and Applications\thanks{Research supported by CE-ITI grant project P202/12/G061 of GA~{\v C}R, GA~UK grant project 1784214, and ERC Starting Grant 306465 (BeyondWorstCase). An extended abstract of these results appeared in the Proceedings of the 25th European Symposium of Algorithms~\cite{KnopEtAl2017b}.}}

\date{}

\author{
Dušan Knop\thanks{\texttt{knop@kam.mff.cuni.cz} Department of Applied Mathematics, Charles University, Prague, Czech Republic \emph{and} Department of Informatics, University of Bergen, Bergen, Norway}
\and
Martin Koutecký\thanks{\texttt{koutecky@technion.ac.il} Department of Applied Mathematics, Charles University, Prague, Czech Republic \emph{and} Technion - Israel Institute of Technology, Haifa, Israel}
\and
Matthias Mnich\thanks{\texttt{mmnich@uni-bonn.de} Institut f{\"u}r Informatik, Universit\"at Bonn, Bonn, Germany \emph{and} Department of Quantitative Economics, Maastricht University, Maastricht, The~Netherlands}
}

\maketitle
\begin{abstract}
  Many fundamental \NP-hard problems can be formulated as integer linear programs (ILPs).
A~famous algorithm by Lenstra solves ILPs in time that is exponential only in the dimension of the program, and polynomial in the size of the ILP.
  That algorithm became a ubiquitous tool in the design of fixed-parameter algorithms for \NP-hard problems, where one wishes to isolate the hardness of a problem by some parameter.
  However, in many cases using Lenstra's algorithm has two drawbacks:
  First, the run time of the resulting algorithms is often doubly-exponential in the parameter, and second, an ILP formulation in small dimension cannot easily express problems involving many different costs.

  Inspired by the work of Hemmecke, Onn and Romanchuk~[Math. Prog. 2013], we develop a single-exponential algorithm for so-called \emph{combinatorial $n$-fold integer programs}, which are remarkably similar to prior ILP formulations for various problems, but unlike them, also allow variable dimension.
  We then apply our algorithm to a few representative problems like \textsc{Closest String}, \textsc{Swap Bribery}, \textsc{Weighted Set Multicover}, and obtain exponential speedups in the dependence on the respective parameters, the input size, or both.

  Unlike Lenstra's algorithm, which is essentially a bounded search tree algorithm, our result uses the technique of augmenting steps.
  At its heart is a deep result stating that in combinatorial $n$-fold IPs, existence of an augmenting step implies existence of a ``local'' augmenting step, which can be found using dynamic programming.
  Our results provide an important insight into many problems by showing that they exhibit this phenomenon, and highlights the importance of augmentation techniques.
\end{abstract}

\sloppy

\section{Introduction}\label{sec:introduction}
The \textsc{Integer Linear Programming} (ILP) problem is fundamental as it models many combinatorial optimization problems.
Since it is \NPc, we naturally ask about the complexity of special cases.
A~fundamental algorithm by Lenstra from 1983 shows that ILPs can be solved in polynomial time when their number of variables (the dimension) $d$ is fixed~\cite{Lenstra1983}; that algorithm is thus a natural tool to prove that the complexity of some special cases of other \NPh problems is also polynomial.

A~systematic way to study the complexity of ``special cases'' of \NPh problems was developed in the past 25 years in the field of parameterized complexity.
There, the problem input is augmented by some integer parameter~$k$, and one then measures the problem complexity in terms of both the instance size $n$ as well as $k$.
Of central importance are algorithms with run times of the form $f(k) n^{\OhOp{1}}$ for some computable function $f$, which are called \emph{fixed-parameter algorithms}; the key idea is that the degree of the polynomial does not grow with $k$.
For background on parameterized complexity, we refer to the monograph~\cite{CyganEtAl2015}.

Kannan's improvement~\cite{Kannan1987} of Lenstra's algorithm runs in time $d^{\OhOp{d}} n$, which is thus a fixed-parameter algorithm for parameter $d$.
Gramm et al.~\cite{GrammEtAl2003} pioneered the application of Lenstra's and Kannan's algorithm in parameterized complexity: they modeled {\sc Closest String} with $k$ input strings as an ILP of dimension $k^{\OhOp{k}}$, and thereby concluded with the first fixed-parameter algorithm for \textsc{Closest String}.
This success led Niedermeier~\cite{Niedermeier2004} to propose in his book:
\begin{quote}
    \emph{[...] It remains to investigate further examples besides
        \textsc{Closest String} where the described ILP approach turns out
        to be applicable. More generally, it would be interesting to
        discover more connections between fixed-parameter algorithms and
        (integer) linear programming.}
\end{quote}
Since then, many more applications of Lenstra's and Kannan's algorithm for parameterized problems have been proposed.
However, essentially all of them~\cite{BredereckEtAl2015,DornSchlotter2012,FellowsEtAl2008,HermelinRozenberg2015,Lampis2012,MnichWiese2015} share a common trait with the algorithm for \textsc{Closest String}: they have a doubly-exponential run time dependence on the parameter.
Moreover, it is difficult to find ILP formulations with small dimension for problems whose input contains many objects with varying cost functions, such as in \textsc{Swap Bribery}~\cite[Challenge \#2]{BredereckEtAl2014}.

\subsection{Our contributions}
We show that a certain form of ILP, which is closely related to the previously used formulations for \textsc{Closest String} and other problems, can be solved in single-exponential time and in variable dimension.
For example, Gramm et al.'s~\cite{GrammEtAl2003} algorithm for \textsc{Closest String} runs in time $2^{2^{\OhOp{k \log k}}} \Oh(\log L)$ for $k$ strings of length~$L$ and has not been improved since 2003, while our algorithm runs in time $k^{\OhOp{k^2}} \OhOp{\log L}$.
Moreover, our algorithm has a strong combinatorial flavor and is based on different notions than are typically encountered in parameterized complexity, most importantly augmenting steps.

As an example of our form of ILP, consider the following ILP formulation of the \textsc{Closest String} problem.
We are given $k$ strings $s_1, \dots, s_k$ of length~$L$ that come (after some preprocessing) from alphabet $[k]:=\{1,\hdots,k\}$, and an integer $d$.
The goal is to find a string $y \in [k]^L$ such that, for each $s_i$, the Hamming distance $d_H(y, s_i)$ is at most $d$, if such $y$ exists.
For $i \in [L]$, $(s_1[i], \dots, s_k[i])$ is the $i$-th column of the input.
Clearly there are at most $k^k$ different column types in the input, and we can represent the input succinctly with multiplicities $b^{\vef}$ of each column type $\vef \in [k]^k$.
Moreover, there are $k$ choices for the output string $y$ in each column.
Thus, we can encode the solution by, for each column type $\vef \in [k]^k$ and each output character $e \in [k]$, describing how many solution columns are of type $(\vef, e)$.
This is the basic idea behind the formulation of Gramm et al.~\cite{GrammEtAl2003}, as depicted on the left:
\[
\begin{array}{rcl|rclr}
\displaystyle \sum_{e \in [k]} \sum_{\vef \in [k]^k} d_H(e, f_j) x_{\vef,e} & \leq & d
&\displaystyle \sum_{\vef \in [k]^k} \sum_{(\vef', e) \in [k]^{k+1}}  d_H(e, f_j) x_{\vef',e}^{\vef} &\leq & d
& \forall j \in [k] \\
\displaystyle \sum_{e \in [k]} x_{\vef,e} &= & b^{\vef}
&\displaystyle \sum_{(\vef',e) \in [k]^{k+1}} x_{\vef',e}^{\vef} &=& b^{\vef}
&\forall \vef \in [k]^k \\
x_{\vef,e} &\geq & 0 &&&& \forall (\vef,e) \in  [k]^{k+1} \\
&&& x_{\vef,e}^{\vef'} &=& 0 & \forall \vef' \neq \vef, \forall e \in [k]\\
&&&
0 \leq x_{\vef,e}^{\vef} &\leq & b^{\vef} & \forall \vef \in [k]^{k} 
\end{array}
\]
\noindent
Let $(1~\cdots~1) = \mathbf{1}^{\transpose}$ be a row vector of all ones.
Then we can view the above as
\[
\begin{array}{ccccl|cccccl}
D_1    & D_2    & \cdots & D_{k^k} & \leq d  & \hspace{9.5em}~
& D & D & \cdots & D & \leq d
\\
\mathbf{1}^{\transpose}    & 0      & \cdots & 0 & = b^1 & \hspace{20em}   
& \mathbf{1}^{\transpose} & 0 & \cdots & 0 & = b^1 
\\
0      & \mathbf{1}^{\transpose}    & \cdots & 0 & = b^2 & \hspace{20em} 
& 0      & \mathbf{1}^{\transpose}    & \cdots & 0 & = b^2
\\
\vdots & \vdots & \ddots & \vdots &  =  \vdots & \hspace{20em}
& \vdots & \vdots & \ddots & \vdots &  =  \vdots
\\
0      & 0      & \cdots & \mathbf{1}^{\transpose} & = b^{k^k} & \hspace{20em}  
& 0      & 0      & \cdots & \mathbf{1}^{\transpose} & = b^{k^k},
\\
\end{array}
\]
where $D = (D_1~D_2~\dots~D_{k^k})$.
The formulation on the right is clearly related to the one on the left, but contains ``dummy'' variables which are always zero. 
This makes it seem unnatural at first, but notice that it has the nice form
\begin{align}
\min\left\{f(\vex)\, \mid E^{(n)}\vex=\veb\,,\ \vel\leq\vex\leq\veu\,,\ \vex\in\Z^{nt}\right\},\label{eq:standard_nfold}\\
~~ \mbox{where }~E^{(n)}:=
\left(
\begin{array}{cccc}
D    & D    & \cdots & D    \\
A    & 0      & \cdots & 0      \\
0      & A    & \cdots & 0      \\
\vdots & \vdots & \ddots & \vdots \\
0      & 0      & \cdots & A    \\
\end{array}
\right).\notag
\end{align}
Here, $r,s,t,n \in \N$, $\veu, \vel \in \Z^{nt}$, $\veb \in \Z^{r+ns}$ and $f: \Z^{nt} \to \Z$ is a separable convex function (i.e., $f(\vex) = \sum_{i=1}^n \sum_{j=1}^t f^i_j(x^i_j)$ with every $f^i_j: \Z \to \Z$ univariate convex),
$E^{(n)}$ is an $(r+ns)\times nt$-matrix, $D \in \Z^{r \times t}$ is an $r\times t$-matrix and $A \in \Z^{s \times t}$ is an $s\times t$-matrix.
We call $E^{(n)}$ the \emph{$n$-fold product of $E = \left(\begin{smallmatrix}D\\A\end{smallmatrix}\right)$}.
By $L = \l \veb, \vel, \veu, f \r$ we denote the length of the binary encoding of the vectors $\veb, \vel, \veu$ and the objective function $f$, where $\l f \r = \l \max_{\vex: \vel \leq \vex \leq \veu} |f(\vex)| \r$ is the encoding length of the maximum absolute value attained by $f$ over the feasible region.
This problem~\eqref{eq:standard_nfold} is known as \emph{$n$\hy fold integer programming} $(IP)_{E^{(n)},\veb,\vel,\veu,f}$.
Building on a dynamic program of Hemmecke, Onn and Romanchuk~\cite{HemmeckeEtAl2013} and a so-called proximity technique of Hemmecke, K{\"o}ppe and Weismantel~\cite{HemmeckeEtAl2014}, Knop and Kouteck{\'y}~\cite{KnopKoutecky2017} prove that:
\begin{proposition}[{\cite[Thm. 7]{KnopKoutecky2017}}]\label{thm:nfold}
There is an algorithm that solves\footnote{Given an IP, we say that to \emph{solve} it is to either (i) declare it infeasible or unbounded or (ii) find a minimizer of it.} \mbox{$(IP)_{E^{(n)},\veb,\vel,\veu,f}$} encoded with $L = \l \veb, \vel, \veu, f \r$ bits in time \mbox{$\Delta^{O(trs + t^2s)}\cdot n^3L$}, where \mbox{$\Delta = 1 + \max\{\|D\|_\infty, \|A\|_\infty\}$}.
\end{proposition}
However, since the ILP on the right for {\sc Closest String} satisfies $t = k^k$, applying Proposition~\ref{thm:nfold} gives no advantage over applying Lenstra to solve the {\sc Closest String} problem.

We overcome this impediment by harnessing the special structure of the ILP for {\sc Closest String}.
Observe that its constraint matrix $A$ has the form
\begin{eqnarray}
\label{eqn:combnfoldip-cond1}
   A = (1~\cdots~1) &    = & \mathbf{1}^{\transpose} \in \Z^{1 \times t}.
\end{eqnarray}
Moreover, under suitable assumptions on the objective function $f$, we call any such IP a \emph{combinatorial $n$\hy{}fold IP}:
\begin{definition}[Combinatorial $n$-fold IP]
  \label{combinatorialNFoldDef}
  Let $A = (1~\cdots~1)\in\mathbb Z^{1\times t}$, let $D\in \Z^{r\times t}$ be a~matrix, and let \mbox{$E = \left(\begin{smallmatrix}D\\A\end{smallmatrix}\right)$}.
  Let $f\colon \Z^{nt} \to \Z$ be a~separable convex function represented by an evaluation oracle.
  A \emph{combinatorial $n$\hy{}fold IP} is 
  \begin{equation}
  \label{eqn:combnfold}
    \min \left\{ f(\vex) \mid \A \vex = \veb, \vel \leq \vex \leq \veu, \vex \in \Z^{nt} \right\} \enspace .
  \end{equation}
\end{definition}
To use $f$ algorithmically, we also want $f$ to admit an efficient optimization oracle for the continuous relaxation of~\eqref{eqn:combnfold}.
This property we can often assume.
Indeed, if for any $\alpha \in \RR$ there is an efficient separation oracle for the \emph{level set} $\{\vex \mid f(\vex) \leq \alpha\}$, then the~continuous optimum $\hat{\vex}$ of $\min \big\{  f(\vex) \mid \A \vex = \veb, \vel \leq \vex \leq \veu \big\}$ can be found in polynomial time using the~ellipsoid method.

Our main result is a fast algorithm for combinatorial $n$-fold IPs, which is exponentially faster in $t$ than previous works for general $n$-fold IPs.
\newcommand{\mainresult}{There is an algorithm that solves any combinatorial $n$-fold IP $(IP)_{E^{(n)},\veb, \vel, \veu, f}$ of size $L = \l \veb, \vel, \veu, f \r$ in time $t^{\OhOp{r}} (\Delta r)^{\OhOp{r^2}} \OhOp{n^3 L} + \mathsf{oo}$, where $\Delta = 1 + \|D\|_\infty$ and $\mathsf{oo}$ is the time required for one call to an optimization oracle for the continuous relaxation of $(IP)_{E^{(n)},\veb, \vel, \veu, f}$.}
\begin{theorem}
\label{thm:combinatorialnfold}
  \mainresult
\end{theorem}
Observe that, when applicable, our algorithm is not only asymptotically faster than Lenstra's, but works even if $n$ is variable (not parameter).

\subsection{Succinctness}
A~common aspect shared by all of our applications is that bounding some parameter of the instance makes it preferable to view the instance in a succinct way (following the terminology of Faliszewski et al.~\cite{FaliszewskiEtAl2006}; Onn~\cite{Onn2014,OnnSarrabezolles2015} calls these problems \emph{huge} whereas Goemans and Rothvo{\ss}~\cite{GoemansRothvoss2014} call them \emph{high multiplicity}).
The \emph{standard} way of viewing an instance is that the input is a collection of individual objects (bricks, matrix columns, voters, covering sets etc.).
The \emph{succinct} way of viewing an instance is by saying that identical objects are of the same \emph{type}, giving a bound $T$ on the number of distinct types, and then presenting the input as numbers $n_1, \dots, n_{T}$ such that $n_i$ is the number of objects of type $i$.
Clearly, any standard instance can be converted to a succinct instance of roughly the same size (the number of objects is an upper bound on $T$), but the converse is not true as the numbers $n_i$ might be large.
Also, it is sometimes not trivial (see Sect.~\ref{subsec:hugenfold}) that the output can be represented succinctly; still, in all cases which we study it can.

In our applications we always state what are the types and what is the upper bound $T$ on the number of types; we assume some arbitrary enumeration of the types.
We also assume that the input is presented succinctly and thus we do not include the time needed to read and convert a standard instance into a succinct instance in the runtime of our algorithms.

\subsection{Applications}
We apply Theorem~\ref{thm:combinatorialnfold} to several fundamental combinatorial optimization problems, for which we obtain exponential improvements in the parameter dependence, the input length, or both.
For a summary of results, see Table~\ref{tab:results}; this list is not meant to be exhaustive.
In fact, we believe that for any Lenstra-based result in the literature which only achieves double-exponential run times, there is a good chance that it can be sped up using our algorithm.
The only significant obstacle seem to be large coefficients in the constraint matrix or an exponential number of ``global'' constraints.

\begin{table}[htb]
  \centering
  \resizebox{\textwidth}{!}{%
    \begin{tabular}{lll}
      \toprule
      Problem                  & Previous best run time & Our result\\
      \toprule
      {\sc Closest String}          & $2^{2^{\OhOp{k \log k}}} \OhOp{\log L}$~\cite{GrammEtAl2003} & $k^{\Oh(k^2)} \OhOp{\log L}$\\
      \midrule
      {\sc Optimal Consensus}       & \FPT for $k \leq 3$, open for $k \geq 4$~\cite{AmirEtAl2011} & $k^{\OhOp{k^2}} \OhOp{\log L}$\\
      \midrule
      Score-{\sc Swap Bribery}      & $2^{2^{\OhOp{|C| \log |C|}}} \OhOp{\log |V|}$~\cite{DornSchlotter2012} & $|C|^{\OhOp{|C|^2}} \OhOp{T^3 \log |V|}$,\\
            & $|C|^{\OhOp{|C|^6}}\OhOp{|V|^3}$~\cite{KnopEtAl2017} & with $T \leq |V|$\\
      \midrule
      C1-{\sc Swap Bribery}         & $2^{2^{\OhOp{|C| \log |C|}}} \OhOp{\log |V|}$~\cite{DornSchlotter2012} & $|C|^{\OhOp{|C|^4}} \OhOp{T^3 \log |V|}$,\\
      & $|C|^{\OhOp{|C|^6}} \OhOp{|V|^3}$~\cite{KnopEtAl2017} & with $T \leq |V|$\\
      \midrule
      {\sc Weighted Set Multicover} & $2^{2^{\OhOp{k \log k}}} \Oh(n)$~\cite{BredereckEtAl2015} & $k^{\OhOp{k^2}} \Oh(\log n)$\\
      \midrule
      {\sc Huge $n$-fold IP}        & \FPT with $D=I$ and & \FPT with parameter-\\
                                    & $A$ totally unimodular & sized domains\\
      \bottomrule
    \end{tabular}}
    \caption{Run time improvements for a few representative problems resulting from this work.}
    \label{tab:results}
\end{table}

\paragraph{Stringology.} A~typical problem from stringology is to find a string $y$ satisfying certain distance properties with respect to $k$ strings $s_1, \dots, s_k$.
All previous fixed-parameter algorithms for such problems we are aware of for parameter~$k$ rely on Lenstra's algorithm, or their complexity status was open (e.g., the complexity of \textsc{Optimal Consensus}~\cite{AmirEtAl2011} was unknown for all $k \geq 4$).
Interestingly, Boucher and Wilkie~\cite{BoucherWilkie2010} show the counterintuitive fact that \textsc{Closest String} is easier to solve when $k$ is large, which makes the parameterization by $k$ even more significant.
Finding an algorithm with run time only single-exponential in~$k$ was a repeatedly posed problem, e.g. by Bulteau et al.~\cite[Challenge \#1]{BulteauEtAl2014} and Avila et al.~\cite[Problem 7.1]{AvilaEtAl2006}.
By applying our result, we close this gap for a wide range of problems.
\begin{theorem}
\label{thm:strings_singleexp}
  The problems
  \begin{itemize}
    \item \textsc{Closest String}, \textsc{Farthest String}, \textsc{Distinguishing String Selection}, \textsc{Neighbor String}, \textsc{Closest String with Wildcards}, \textsc{Closest to Most Strings}, \textsc{$c$\hy HRC} and \textsc{Optimal Consensus} are solvable in time $k^{\Oh(k^2)} \Oh(\log L)$, and,
    \item \textsc{$d$-Mismatch} is solvable in time $k^{\Oh(k^2)} \Oh(L^2 \log L)$, 
  \end{itemize}
  for inputs consisting of $k$ strings of length $L$ succinctly encoded by multiplicities of identical columns.
\end{theorem}

\paragraph{Computational social choice.}
A~typical problem in computational social choice takes as input an election consisting of a set $V$ of voters and a set $C$ of candidates which are ranked by the voters; the objective is to manipulate the election in certain ways to let a desired candidate win the election under some voting rule $\mathcal R$.
This setup leads to a class of bribery problems, a prominent example of which is $\mathcal R$-{\sc Swap Bribery} where manipulation is by swaps of candidates which are consecutive in voters' preference orders.
For a long time, the only known algorithms minimizing the number of swaps required run times which were doubly-exponential in $|C|$; improving those run times was posed as a challenge~\cite[Challenge \#1]{BredereckEtAl2014}.
Recently, Knop et al.~\cite{KnopEtAl2017} solved the challenge using Proposition~\ref{thm:nfold}.
However, Knop et al.'s result has a cubic dependence $O(|V|^3)$ on the number of voters, and the dependence on the number of candidates is still quite large, namely $|C|^{\Oh(|C|^6)}$.

We improve their result to \emph{logarithmic} dependence on $|V|$, and smaller dependence on $|C|$.
By $T$ we denote the number of \emph{voter types}, where two voters are of the same type if they have the same preferences over the candidates and the same cost function for bribery; clearly $T \leq |V|$.
\begin{theorem}
\label{thm:applicationBribery}
  $\R$-\textsc{Swap Bribery} can be solved in time
  \begin{itemize}
    \item $|C|^{\Oh(|C|^2)} \Oh(T^3 (\log |V| + \log \sigma_{\max}))$ for $\R$ any natural scoring protocol, and,
    \item $|C|^{\Oh(|C|^4)} \Oh(T^3 (\log |V| + \log \sigma_{\max}))$ for $\R$ any C1 rule,
  \end{itemize}
  where $T$ is the number of voter types and $\sigma_{\max}$ is the maximum cost of a swap.
\end{theorem}

\paragraph{Connections between stringology and computational social choice.}
Challenge \#3 of Bulteau et al.~\cite{BulteauEtAl2014} asks for connections between problems in stringology and computational social choice.
We demonstrate that in both fields combinatorial $n$-fold IP is an important tool.
An important feature of both {\sc Bribery}-like problems and {\sc Closest String}-like problems is that permuting voters or characters does not have any effect.
This fits well the $n$-fold IP format, which does not allow any interaction between bricks.
It seems that this feature is important, as when it is taken away, such as in {\sc Closest Substring}, the problem becomes $\mathsf{W}[1]$-hard~\cite{Marx2008}, even for parameter $d+k$.

Another common feature is that both types of problems naturally admit ILP formulations for succinct variants of the problems, as mentioned above.
Moreover, it was precisely this fact that made all previous algorithms doubly-exponential---the natural succinct formulation has exponentially many (in the parameter) variables and thus applying Lenstra's algorithm leads to a doubly-exponential runtime.

\paragraph{Weighted set multicover.}
Bredereck et al.~\cite{BredereckEtAl2015} define the \textsc{Weighted Set Multicover} (WSM) problem, which is a significant generalization of the classical \textsc{Set Cover} problem.
Their motivation to study WSM was that it captures several problems from computational social choice and optimization problems on graphs implicit to previous works~\cite{FialaEtAl2017,GajarskyEtAl2013,Lampis2012}.
Bredereck et al.~\cite{BredereckEtAl2015} design an algorithm for WSM that runs in time~$2^{2^{\Oh(k\log k)}}\Oh(n)$, using Lenstra's algorithm.

Again, our result yields an \emph{exponential} improvement of that by Bredereck et al.~\cite{BredereckEtAl2014}, both in the dependence on the parameter and the size of the instance:
\begin{theorem}
\label{thm:weightedSetMulticover}
  There is an algorithm that solves \textsc{Weighted Set Multicover} in time \mbox{$k^{\Oh(k^2)} \Oh(\log n + \log w_{\max})$} for succinctly represented instances of $n$ sets over a universe of size $k$, where $w_{\max}$ is the maximum weight of any set.
\end{theorem}

\paragraph{Huge $n$-fold IP.}
Onn~\cite{Onn2014} introduces a high-multiplicity version of the standard $n$-fold IP problem~\eqref{eq:standard_nfold}, where the number of bricks is now given in binary.
It thus closely relates to the {\sc Cutting Stock} problem, the high-multiplicity version of {\sc Bin Packing} where the number of items for given size is given in binary; the complexity of {\sc Cutting Stock} for constantly many item sizes was a long-standing open problem that was recently shown to be polynomial-time solvable by Goemans and Rothvo{\ss}~\cite{GoemansRothvoss2014} and by Jansen and Klein~\cite{JansenK2017}.
Previously, \textsc{Huge $n$-fold IP} was shown to be fixed-parameter tractable when $D=I$ and $A$ is totally unimodular; using our result, we show that it is also fixed-parameter tractable when $D$ and~$A$ are arbitrary, but the size of variable domains is bounded by a parameter.


\subsection{Comparison with Lenstra's algorithm}
The basic idea behind Lenstra's algorithm is the following.
Given a system $A \vex \leq \veb$ it is possible to compute its volume and determine that it is either too large not to contain an integer point, or too small not to be flat in some direction.
In the first case we are done; in the second case we can take $d$ slices of dimension $d-1$ and recurse into them, achieving a $d^{\Oh(d)} n^{\Oh(1)}$ runtime.
Note that we only decide feasibility; optimization can be then done by binary search.
On the other hand, the basic idea behind our algorithm is the following.
We only focus on optimizing and later show that testing feasibility reduces to it.
Starting from some feasible solution, the crucial observation is that if there is a step improving the objective, there is one which does not modify many variables, and can be found quickly by dynamic programming.
Moreover, if the current solution is far from the optimum, then it is possible to make a long step, and polynomially many long steps will reach the optimum.

More concretely, consider the run of these two algorithms on an instance of {\sc Closest String} consisting in $k$ strings each of length $L$.
Lenstra's algorithm essentially either determines that the bounds are loose enough that there must exist a solution, or (oversimplifying) determines that there is a column type $\vef \in [k]^k$ and a character $e \in [k]$ such that there are at most $k^k$ consecutive choices for how many times the solution contains character $e$ at a column of type $\vef$.
Then, we recurse, obtaining a $2^{2^{\Oh(k \log k)}}\Oh(\log L)$-time algorithm.
On the other hand, our algorithm views the problem as an optimization problem, so we think of starting with a string of all blanks which is trivially at distance~$0$ from any string, and the goal is to fill in all blanks such that the result is still in distance at most $d$ from the input strings.
An augmenting step is a set of character swaps that decreases the number of blanks.
The crucial observation is that if an augmenting step exists, then there is also one only changing few characters, and it can be found in time $k^{\Oh(k^2)}\Oh(\log L)$.
Thus (omitting details), we can iteratively find augmenting steps until we reach the optimum.

\paragraph{Related work.}
\label{sec:relatedwork}
Our main inspiration are augmentation methods based on Graver bases, especially a fixed-parameter algorithm for $n$-fold IP of Hemmecke, Onn and Romanchuk~\cite{HemmeckeEtAl2013}.
Our result improves the runtime of their algorithm for a special case.
All the following related work is orthogonal to ours in either the achieved result, or the parameters used for it.

In fixed dimension, Lenstra's algorithm~\cite{Lenstra1983} was generalized for arbitrary convex sets and quasiconvex objectives by Khachiyan and Porkolab~\cite{KhachiyanPorkolab2000}.
The currently fastest algorithm of this kind is due to Dadush et al.~\cite{DadushEtAl2011}.
The first notable fixed-parameter algorithm for a non-convex objective is due to Lokshtanov~\cite{Lokshtanov2015}, who shows that optimizing a quadratic function over the integers of a polytope is fixed-parameter tractable if all coefficients are small.
Ganian and Ordyniak~\cite{GanianOrdyniak2016} and Ganian et al.~\cite{GanianEtAl2017} study the complexity of ILP with respect to structural parameters such as treewidth and treedepth, and introduce a new parameter called \emph{torso-width}.

Besides fixed-parameter tractability, there is interest in the (non)existence of kernels of ILPs, which formalize the (im)possibility of various preprocessing procedures.
Jansen and Kratsch~\cite{JansenKratsch2015} show that ILPs containing parts with simultaneously bounded treewidth and bounded domains are amenable to kernelization, unlike ILPs containing totally unimodular parts.
Kratsch~\cite{Kratsch2016} studies the kernelizability of sparse ILPs with small coefficients.

\section{Preliminaries}
\label{sec:preliminaries}
For positive integers $m,n$ we set $[m\,:\,n] = \{m,\ldots, n\}$ and $[n] = [1\,:\,n]$.
For a graph~$G$ we denote by $V(G)$ its set of vertices.
We write vectors in boldface (e.g., $\vex, \vey$) and their entries in normal font (e.g., the $i$-th entry of~$\vex$ is~$x_i$).
If~$A$ is a matrix, $A_r$ denotes its $r$-th column.
%
For a matrix $A \in \Z^{m \times n}$, vectors $\veb \in \Z^{m}$, $\vel, \veu \in \Z^{n}$ and a function $f: \Z^n \to \Z$, let $(IP)_{A, \veb, \vel, \veu, f}$ be the problem
\[
  \min\left\{f(\vex)\, \mid A\vex=\veb\,,\ \vel\leq\vex\leq\veu\,,\ \vex\in\Z^{n}\right\}.
\]
We say that $\vex$ is \emph{feasible} for $(IP)_{A, \veb, \vel, \veu, f}$ if $A\vex = \veb$ and $\vel \leq \vex \leq \veu$.

\paragraph{Graver bases and augmentation.}
Let us now introduce Graver bases and discuss how they can be used for optimization.
We also recall $n$-fold IPs; for background, we refer to the books of Onn~\cite{Onn2010} and De Loera et al.~\cite{DeLoeraEtAl2013}.

Let $\vex, \vey$ be $n$-dimensional integer vectors.
We call $\vex, \vey$ \emph{sign\hy{}compatible} if they lie in the same orthant, that is, for each $i \in [n]$ the sign of $x_i$ and $y_i$ is the same.
We call $\sum_i \veg^i$ a \emph{sign\hy{}compatible sum} if all $\veg^i$ are pair-wise sign\hy{}compatible.
Moreover, we write $\vey \sqsubseteq \vex$ if $\vex$ and $\vey$ are sign\hy{}compatible and $|y_i| \leq |x_i|$ for each $i \in [n]$, and write $\vey \sqsubset \vex$ if at least one of the inequalities is strict.
Clearly, $\sqsubseteq$ imposes a partial order called ``conformal order'' on $n$-dimensional vectors.
For an integer matrix $A \in \Z^{m \times n}$, its \emph{Graver basis} $\G(A)$ is the set of $\sqsubseteq$-minimal non-zero elements of the \emph{lattice} of $A$, $\ker_{\Z}(A) = \{\vez \in \Z^n \mid A \vez = \mathbf{0}\}$.
An important property of~$\G(A)$ is the following.
\begin{proposition}[{\cite[Lemma 3.2]{Onn2010}}]
  \label{prop:graver_conformal_sum}
  Every integer vector $\vex \neq \mathbf{0}$ with $A \vex = \mathbf{0}$ is a sign\hy{}compatible sum $\vex = \sum_i \veg^i$ of Graver basis elements $\veg^i \in \G(A)$, with some elements possibly appearing with repetitions.
\end{proposition}

Let $\vex$ be a feasible solution to $(IP)_{A, \veb, \vel, \veu, f}$.
We call $\veg$ a \emph{feasible step} if $\vex + \veg$ is feasible for $(IP)_{A, \veb, \vel, \veu, f}$.
Further, call a feasible step $\veg$ \emph{augmenting} if $f(\vex+\veg) < f(\vex)$.
An augmenting step $\veg$ and a \emph{step length} $\alpha \in \Z$ form an \emph{$\vex$\hy{}feasible step pair} with respect to a feasible solution $\vex$ if $\vel \le \vex + \alpha\veg \le \veu$.
An augmenting step~$\veg$ and a step length $\alpha \in \Z$ form a \emph{Graver\hy best step} if $f(\vex + \alpha \veg) \leq f(\vex + \alpha' \tilde{\veg})$ for all $\vex$\hy{}feasible step pairs $(\tilde{\veg},\alpha') \in \G(A)\times \Z$.

The \emph{Graver\hy{}best augmentation procedure} for $(IP)_{A, \veb, \vel, \veu, f}$ with given feasible solution $\vex_0$ works as follows:
\begin{enumerate}
  \item \label{step1}If there is no Graver\hy{}best step for $\vex_0$, return it as optimal.
  \item If a Graver\hy{}best step $(\alpha, \veg)$ for $\vex_0$ exists, set $\vex_0 := \vex_0 + \alpha \veg$ and go to \ref{step1}. 
\end{enumerate}
\begin{proposition}[{\cite[implicit in Theorem 3.4.1]{DeLoeraEtAl2013}}]
\label{prop:graverbest}
  Given a feasible solution~$\vex_0$ for $(IP)_{A, \veb, \vel, \veu, f}$ where $f$ is separable convex, the Graver\hy{}best augmentation procedure finds an optimum of $(IP)_{A, \veb, \vel, \veu, f}$ in at most $2n-2 \log M$ steps, where $M = f(\vex_0) - f(\vex^*)$ and $\vex^*$ is any minimizer of $f$.
\end{proposition}

\paragraph{\texorpdfstring{$n$-fold}{n-fold} IP.}
The structure of $E^{(n)}$ (in problem~\eqref{eq:standard_nfold}) allows us to divide the $nt$ variables of $\vex$ into $n$ \textit{bricks} of size~$t$.
We use subscripts to index within a brick and superscripts to denote the index of the brick, i.e.,~$x_j^i$ is the $j$-th variable of the $i$-th brick with $j \in [t]$ and $i \in [n]$.

\section{Combinatorial \texorpdfstring{$\boldsymbol{n}$-fold}{n-fold} IPs}
\label{sec:maintheorem}
This section is dedicated to proving Theorem~\ref{thm:combinatorialnfold}.
We fix an instance of combinatorial $n$\hy fold IP, that is, a tuple $(n, D, \veb, \vel, \veu, f)$.


\subsection{Graver complexity of combinatorial $n$-fold IP}
The key property of the $n$-fold product $E^{(n)}$ is that, for any $n \in \N$, the number of nonzero bricks of any $\veg \in \cGEn$ is bounded by some constant $g(E)$ called the \emph{Graver complexity of~$E$}.
A~proof is given for example by Onn~\cite[Lemma 4.3]{Onn2010}; it goes roughly as follows.
Consider any $\veg \in \cGEn$ and take its restriction to its nonzero bricks $\bar{\veg}$.
By Proposition~\ref{prop:graver_conformal_sum}, each brick $\bar{\veg}^j$ can be decomposed into elements from~$\G(A)$, giving a vector $\veh$ whose bricks are elements of $\G(A)$.
Then, consider a compact representation $\vev$ of~$\veh$ by counting how many times each element from $\G(A)$ appears.
Since $\veg \in \cGEn$ and $\veh$ is a decomposition of its nonzero bricks, we have that $\sum_{j} D \veh^j = 0$.
Let $G$ be a matrix with the elements of $\G(A)$ as columns.
It is not difficult to show that $\vev \in \G(DG)$.
Since $\|\vev\|_1$ is an upper bound on the number of bricks of~$\veh$ and thus of nonzero bricks of $\veg$ and clearly does not depend on $n$, \mbox{$g(E) = \max_{v \in \G(DG)} \|\vev\|_1$} is finite.

Let us make precise two observations from this proof.
\begin{lemma}[{\cite[Lemma 3.1]{HemmeckeEtAl2011}, \cite[implicit in proof of Lemma 4.3]{Onn2010}}]
\label{lem:partialsums}
  Let $(\veg^1, \dots, \veg^n) \in \cGEn$.
  Then for $i = 1,\hdots,n$ there exist vectors $\veh^{i,1}, \dots, \veh^{i,n_i} \in \G(A)$ such that $\veg^i = \sum_{k=1}^{n_i} \veh^{i,k}$, and $\sum_{i=1}^n n_i \leq g(E)$.
\end{lemma}

\begin{lemma}[{\cite[Lemma 6.1]{HemmeckeEtAl2011}, \cite[implicit in proof of Lemma 4.3]{Onn2010}}]
\label{lem:ge_estimate}
  Let $D \in \Z^{r \times t}$, $A \in \Z^{s \times t}$ and let $G \in \Z^{t \times p}$ be the matrix whose columns are the elements of $\G(A)$.
  Then $|\G(A)| \leq \|A\|_\infty^{st}$, and for $E = \left(\begin{smallmatrix}D\\A\end{smallmatrix}\right)$ it holds
  \[
    g(E) \leq \max_{\vev \in \G(DG)} \|\vev\|_1 \leq \|A\|_\infty^{st} \cdot (r \|DG\|_\infty)^r.
  \]
\end{lemma}

Notice that this bound on $g(E)$ is exponential in $t$.
Our goal now is to improve the bound on $g(E)$ in terms of $t$, exploiting the simplicity of the matrix~$A$ in combinatorial $n$-fold IPs.

To see this, we will need to understand the structure of $\G(\mathbf{1}^{\transpose})$:
\begin{lemma}
\label{lem:GA_structure}
  It holds that
  \begin{itemize}
    \item $\G(\mathbf{1}^{\transpose}) = \{\veg \mid \veg \mbox{ has one $1$ and one $-1$ and $0$ otherwise}\} \subseteq \Z^t$,
    \item $p = |\G(\mathbf{1}^{\transpose})| = t(t-1)$,
    \item and $\|\veg\|_1 = 2$ for all $\veg \in \G(\mathbf{1}^{\transpose})$.
  \end{itemize}
\end{lemma}
\begin{proof}
  Observe that the claimed set of vectors is clearly $\sqsubseteq$-minimal in $\ker_{\Z}(\mathbf{1}^{\transpose})$.
  We are left with proving there is no other non-zero $\sqsubseteq$-minimal vector in $\ker_{\Z}(\mathbf{1}^{\transpose})$.
  For contradiction assume there is such a vector $\veh$.
  Since it is non-zero, it must have a positive entry $h_i$.
  On the other hand, since $\mathbf{1}^{\transpose} \veh = \mathbf{0}$, it must also have a negative entry $h_j$.
  But then $\veg$ with $g_i = 1$, $g_j = -1$ and $g_k = 0$ for all $k \not\in \{i,j\}$ is $\veg \sqsubset \veh$, a contradiction.
  The rest follows.
\end{proof}

With this lemma in hand, we can prove that:
\begin{lemma}
\label{lem:ge_bound}
  Let $D \in \Z^{r \times t}$, $E = \left(\begin{smallmatrix}D\\\mathbf{1}^{\transpose}\end{smallmatrix}\right)$, and $\Delta = 1 + \|D\|_\infty$.
  Then, $g(E) \leq t^2 (2 r \Delta)^r$.
\end{lemma}
\begin{proof}
  We simply plug the correct values into the bound of Lemma~\ref{lem:ge_estimate}.
  By Lemma~\ref{lem:GA_structure}, $p = t(t-1) \le t^2$.
  Also, $\|DG\|_\infty \leq \max_{\veg \in \G(\mathbf{1}^{\transpose})} \left\{ \|D\|_\infty \cdot \|\veg\|_1 \right\} \leq 2\Delta$ where the last inequality follows from $\|\veg\|_1 = 2$ for all $\veg \in \G(\mathbf{1}^{\transpose})$, again by Lemma~\ref{lem:GA_structure}.
\end{proof}

\subsection{Dynamic programming}
\label{sec:combNFoldDP}
Hemmecke, Onn and Romanchuk~\cite{HemmeckeEtAl2013} devise a clever dynamic programming algorithm to find augmenting steps for a feasible solution of an $n$-fold IP.
Lemma~\ref{lem:partialsums} is key in their approach, as they continue by building a set $Z(E)$ of all sums of at most $g(E)$ elements of $\G(A)$ and then use it to construct the dynamic program.
However, such a set $Z(E)$ would clearly be of size exponential in $t$---too large to achieve our single-exponential run times.
In their dynamic program, layers correspond to partial sums of elements of~$\G(A)$.

Our insight is to build a different dynamic program.
In our dynamic program, we will exploit the simplicity of $\G(A) = \G(\mathbf{1}^{\transpose})$ so that the layers will immediately correspond to the coordinates $h^i_j$ of an augmenting vector~$\veh$.
Additionally, we also differ in how we enforce feasibility with respect to the upper rows $(D~D~\cdots~D)$.

We now give the details of our approach.
Let $\Sigma(E) = \prod_{j = 1}^r \left[-2\Delta \cdot g(E) \,:\, 2\Delta \cdot g(E) \right]$ be the \emph{signature set of $E$} whose elements are \emph{signatures}.
Essentially, we will use the signature set to keep track of partial sums of prefixes of the augmenting vector $\veh$ to ensure that it satisfies $D \veh = \mathbf{0}$.
Crucially, we notice that to ensure $D \veh = \mathbf{0}$, it suffices to remember the partial sum of the prefixes of $\veh$ multiplied by $D$, thus shrinking them to dimension $r$.
This is another insight which allows us to avoid the exponential dependence on $t$.
Note that $\left|\Sigma(E)\right| \leq \left(1 + 4 g(E) \Delta \right)^r$.

\begin{definition}[Augmentation graph]
\label{def:augmentationGraph}
  Let $\vex$ be a feasible solution for $(IP)_{E^{(n)}, \veb, \vel, \veu, f}$ and let ${\alpha \in \N}$.
  Their \emph{augmentation graph} $DP(\vex, \alpha)$ is a vertex-weighted directed layered graph with two distinguished vertices $S$ and~$T$ called the \emph{source} and the \emph{sink}; and $nt$ layers $\L(1,1),\ldots, \L(n,t)$ structured according to the bricks, such that for all $i \in [n]$, $j \in [t]$, 
  \[
    \L(i,j) = (i,j) \times [-g(E) \, :\, g(E)] \times [-g(E) \, :\, g(E)] \times \Sigma(E).
  \]
  Thus, each vertex is a tuple $\big(i,j, h^i_j, \beta^i_j, \vesigma^i_j\big)$, with the following meaning:
  \begin{itemize}
    \item $i \in [n]$ is the index of the brick,
    \item $j \in [t]$ is the position within the brick,
    \item $h^i_j \in [-g(E) \,:\, g(E)]$ is the value of the corresponding coordinate of a proposed augmenting vector $\veh$,
    \item $\beta^i_j \in [-g(E) \,:\, g(E)]$ is a brick prefix sum  $\sum_{\ell=1}^j h^i_j$ of the proposed augmenting vector $\veh$, and,
    \item $\vesigma^i_j \in \Sigma(E)$ is the signature, representing the prefix sum $\sum_{k=1}^i D \veh^k + \sum_{\ell=1}^j D_\ell h^i_\ell$.
  \end{itemize}
  A vertex $\big(i,j, h^i_j, \beta^i_j, \vesigma^i_j\big)$ has weight $f^i_j(\alpha h^i_j + x^i_j) - f^i_j(x^i_j)$.
 
  Let $S = \big(0,t,0,0, \mathbf{0}\big)$ and $T = \big(n+1,1,0,0, \mathbf{0}\big)$, where the last coordinate is an $r$-dimensional all-zero vector.
 
  \noindent\textit{Edges to the first layer of a brick.}
  Every vertex $\big(i,t, h^i_t, 0, \vesigma^i_t\big)$ has edges to each vertex $\big(i+1,1, h^{i+1}_1, h^{i+1}_1, \vesigma^{i+1}_1\big)$ for which $l^{i+1}_1 \leq x^{i+1}_1 + \alpha h^{i+1}_1 \leq u^{i+1}_1$ and $\vesigma^{i+1}_1 = \vesigma^{i}_t + D_1 h^{i+1}_1$.
  Recall that $D_j$ is the $j$-th column of matrix~$D$.
  We emphasize that there are no other outgoing edges from layer $\L(i,t)$ to layer $\L(i+1,1)$.

  \noindent\textit{Edges within a brick.}
  Every vertex $\big(i,j, h^i_j, \beta^i_j, \vesigma^i_j)$ with $j < t$ has edges to each vertex $\big(i, j+1, h^i_{j+1}, \beta^i_{j+1}, \vesigma^i_{j+1} \big)$ for which
  \begin{itemize}
    \item $l^{i}_{j+1} \leq x^{i}_{j+1} + \alpha h^{i}_{j+1} \leq u^{i}_{j+1}$,
    \item $\beta^{i}_{j+1} = \beta^i_j + h^{i}_{j+1}$, with $\beta^i_{j+1} \in \left[-g(E)\colon g(E)\right]$, and,
    \item $\vesigma^{i}_{j+1} = \vesigma^{i}_{j} + D_{j+1} h^{i}_{j+1}$.
  \end{itemize}
\end{definition}
See Fig.~\ref{fig:DPgrahpSchema} for a scheme of the augmentation graph.

Note that by the bounds on $g(E)$ by Lemma~\ref{lem:ge_bound},
the number of vertices in each layer of $DP(\vex, \alpha)$ is bounded by
\begin{equation}
\label{LmaxBound}
  L_{\max} \leq g(E)^2 \cdot \left|\Sigma\right|
           \leq \left(t^2(2r\Delta)^r\right)^2 \cdot \left(1 + 4\Delta \cdot \left(t^2(2r\Delta)^r\right) \right)^r
           \leq \left(t^2 (2r\Delta)^r\right)^{\Oh(r)} \enspace .
\end{equation}

\begin{figure}[bt]
  \usetikzlibrary{calc,positioning,fit,shapes,decorations.pathreplacing,scopes}

\begin{tikzpicture}
  \tikzstyle{layer}=[rounded corners, inner sep=0pt,draw]
  \tikzstyle{layerLabel}=[rounded corners, inner sep=3pt, fill=gray!40]
  \tikzstyle{state}=[circle, draw, inner sep=1pt, minimum height=.5cm, minimum width=.5cm,align=center]
  \tikzstyle{finalState}=[state, black!30, fill, text=black, draw=black]
  \tikzstyle{edge}=[ultra thick,->,>=stealth]

\begin{scope}[local bounding box=jj+1]
  \node[layerLabel] (LijLabel) {$\ (\,i\,,\,j\,)\ $};
  \node[state,label={[xshift=7pt, yshift=4pt]0:\vdots}] at ($(LijLabel) - (0,2.5)$) (Nij) {$h, \beta, \vesigma$};
  \node at ($(LijLabel) - (0,4.2)$) (Dij) {};
  \node[layer,fit=(LijLabel)(Nij)(Dij)] {};

  \node[layerLabel] at ($(LijLabel) + (3.6,0)$) (Lij+1Label) {$(i,j + 1)$};
  \node[state] at ($(Lij+1Label) - (0,.8)$) (Nij+1) {$\hat{h}, \hat{\beta}, \hat{\vesigma}$};
  \node[state] at ($(Lij+1Label) - (0,3.5)$) (N2ij+1) {$\bar{h}, \bar{\beta}, \bar{\vesigma}$};
  \node at ($(Lij+1Label) - (0,4.2)$) (Dij+1) {};
  \node[layer,fit=(Lij+1Label)(Nij+1)(Dij+1)] (Lij+1) {};

  \draw[edge] (Nij) to node[midway, above, rotate=22]{$\hat{\beta} = \beta + \hat{h}$} node[midway, below, rotate=22]{$\hat{\vesigma} = \vesigma + D^j\hat{h}$}(Nij+1) ;
  \draw[edge] (Nij) to (N2ij+1);
\end{scope}
\node at ($(jj+1) + (0,3)$) {Transitions within brick};

\begin{scope}[local bounding box=ii+1]
  \node[layerLabel] at ($(Lij+1Label) + (3.2,0)$) (LitLabel) {$\ (\,i\,,\,t\,)\ $};
  \node[state,label={[xshift=7pt, yshift=4pt]0:\vdots}] at ($(LitLabel) - (0,2.5)$) (Nit) {$h, \beta, \vesigma$};
  \node at ($(LitLabel) - (0,4.2)$) (Dit) {};
  \node[layer,fit=(LitLabel)(Nit)(Dit)] (Lit) {};

  \node[layerLabel] at ($(LitLabel) + (3.6,0)$) (Li+11Label) {$(i + 1,1)$};
  \node[state] at ($(Li+11Label) - (0,.8)$) (Ni+11) {$\hat{h}, \hat{h}, \hat{\vesigma}$};
  \node[state] at ($(Li+11Label) - (0,3.5)$) (N2i+11) {$\bar{h}, \bar{h}, \bar{\vesigma}$};
  \node at ($(Li+11Label) - (0,4.2)$) (Di+11) {};
  \node[layer,fit=(Li+11Label)(Ni+11)(Di+11)] {};

  \draw[edge] (Nit) to node[midway, below, rotate=22]{$\hat{\vesigma} = \vesigma + D^1\hat{h}$}(Ni+11) ;
  \draw[edge] (Nit) to (N2i+11);
\end{scope}
\node at ($(ii+1) + (0,3)$) {Transitions between bricks};

\node at ($(Lij+1)!.5!(Lit)$) {$\cdots$};
\end{tikzpicture}
  \caption{Transitions in the augmentation graph $DP(\vex, \alpha)$.}
\label{fig:DPgrahpSchema}
\end{figure}
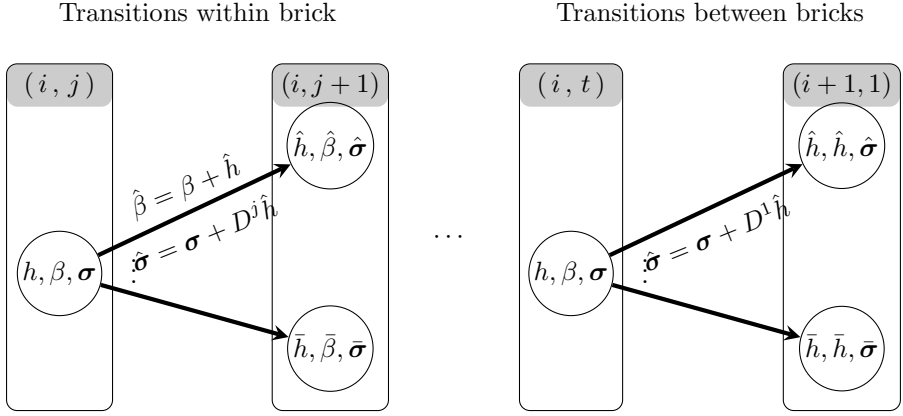

Let $P$ be an $S$--$T$ path in $DP(\vex,\alpha)$.
We define the \emph{$P$-augmentation vector} $\veh \in \Z^{nt}$ by the $h^i_j$-coordinates of the vertices of $P$.

Let $\vex$ be a feasible solution of the combinatorial $n$\hy{}fold IP instance fixed at the beginning of Sect.~\ref{sec:maintheorem}.
We say that $\veh$ is a {\em solution of $DP(\vex,\alpha)$} if there exists an $S$--$T$ path $P$ such that $\veh$ is the $P$-augmentation vector. The weight $w(\veh)$ is then defined as the weight of the path $P$; note that $w(\veh) = f(\vex + \alpha \veh) - f(\vex)$.

The following lemma relates solutions of $DP(\vex, \alpha)$ to potential feasible steps in $\cGEn$.
\begin{lemma}
\label{DP:solutionProperties}
  Let $\vex \in \Z^{nt}$ be a feasible solution, $\alpha \in \N$, and let $\veh$ be a solution of $DP(\vex, \alpha)$.
  Then $\vel \le\vex+\alpha \veh\le\veu$ and $\A\veh=\mathbf{0}$.
\end{lemma}
\begin{proof}
  To see that $\vel \le \vex + \alpha \veh \le \veu$, recall that there is no incoming edge to a vertex $(i,j, h^i_j, \beta^i_j, \vesigma^i_j)$ which would violate the bound $l^i_j \leq x^i_j + \alpha h^i_j \leq u^i_j$.
	
  To see that $\A \veh = \mathbf{0}$, first observe that by the definition of $\beta^i_j$, and the condition that only if $\beta^i_t = 0$ there is an outgoing edge, we have that every brick~$\veh^i$ satisfies $\mathbf{1}^{\transpose} \veh^i = 0$.
  Second, by the definition of $\vesigma^i_j$ and the edges incoming to $T$, we have that $D \veh = \mathbf{0}$.
  Together, this implies $\A \veh = \mathbf{0}$. 
\end{proof}

\begin{lemma}
\label{lem:DPproperties}
  Let $\vex \in \Z^{nt}$ be a feasible solution.
  Then every $\veg \in\cGEn$ with $\vel \leq \vex + \alpha \veg \leq \veu$ is a~solution of $DP(\vex, \alpha)$.
\end{lemma}
\begin{proof}
  Let $\veg\in\cGEn$ satisfy $\vel \leq \vex + \alpha \veg \leq \veu$.
  We shall construct an $S$--$T$ path $P$ in $DP(\vex, \alpha)$ such that~$\veg$ is the $P$-augmentation vector.
  We will describe which vertex is selected from each layer, and argue that this well defined.
  Then, by the definition of $DP(\vex, \alpha)$, it will be clear that the selected vertices are indeed connected by edges.
  
  In layer $\L(1,1)$, we select vertex $\big(1,1,g^1_1, g^1_1, D_1 g^1_1\big)$.
  In layer $\L(i,j)$, we select vertex $\big(i,j, g^i_j, \beta^i_j, \vesigma^i_j \big)$ with $\beta^i_j = \beta^i_{j-1} + g^i_j$ and $\vesigma^i_j = \vesigma^i_{j-1} + D_j g^i_j$ if $j > 1$, and $\beta^i_j = \beta^{i-1}_{t} + g^i_j$ and $\vesigma^i_j = \vesigma^{i-1}_t + g^{i-1}_t$ otherwise.
  
  We shall argue that this is well defined, i.e., that all of the specified vertices actually exist.
  From Lemma~\ref{lem:partialsums} it follows that $\veg$ can be decomposed into $M \leq g(E)$ vectors $\tilde{\veg}^1, \dots \tilde{\veg}^M \in \G(\mathbf{1}^{\transpose})$.
  By Lemma~\ref{lem:GA_structure}, $\|\tilde{\veg}^i\|_1 \leq 2$ for every $i$, which implies that $\|\veg\|_\infty \leq g(E)$.
  Moreover, since $M \leq g(E)$, we also have that for every $i \in [n]$ and every $j \in [t]$, $|\sum_{\ell = 1}^j g^i_j| \leq g(E)$ (i.e., the brick prefix sum is also bounded by $g(E)$ in absolute value). 
  Thus, vertices with the appropriate $g^i_j$- and $\beta^i_j$- coordinates exist.
  Regarding the $\vesigma^i_j$ coordinate, we make a similar observation: for every $i \in [n]$ and $j \in [t]$, $\big (\sum_{\hat{\ell}=1}^{i-1} D \veg^{\hat{\ell}}\big) + \big(\sum_{\ell = 1}^{j} D_\ell g^i_\ell \big) \in \Sigma(E)$.
  
  From the definition of the edges in $DP(\vex, \alpha)$, and the fact that \mbox{$\vel \leq \vex + \alpha \veg \leq \veu$}, the selected vertices create a path.
\end{proof}

\begin{lemma}[optimality certification]
\label{lem:optimality_cert}
  There is an algorithm that, given a feasible solution $\vex\in\Z^{nt}$ for $(IP)_{E^{(n)}, \veb, \vel, \veu, f}$ and $\alpha \in \N$, in time $nt L_{\max}^2 \leq t^{\Oh(r)}(\Delta r)^{\Oh(r^2)} n$ either finds a~vector~$\veh$ such that (i) $\A\veh = \mathbf{0}$, (ii) $\vel\le\vex + \alpha \veh\le\veu$ and (iii) $f(\vex+ \alpha\veh) < f(\vex)$, or decides that none exists.
\end{lemma}
\begin{proof}
  It follows from Lemma~\ref{DP:solutionProperties} that all solutions of $DP(\vex)$ fulfill (i) and (ii). 
  Observe that if we take $\veh$ to be a solution of $DP(\vex, \alpha)$ with minimum weight, then either $f(\vex) = f(\vex+\alpha\veg)$ or $f(\vex) > f(\vex+\alpha\veg)$.
  Due to Lemma~\ref{lem:DPproperties} the set of solutions of $DP(\vex, \alpha)$ contains all $\veg \in \cGEn$ with $\vel \leq \vex + \alpha \veg \leq \veu$.
  Thus, by Proposition~\ref{prop:graverbest}, if $f(\vex) = f(\vex+\alpha \veg)$, no $\veg$ satisfying all three conditions (i), (ii), (iii) exist.

  Our goal is then to find the lightest $S$--$T$ path in the graph $DP(\vex, \alpha)$.
  However, since edges of negative weight will be present, we cannot use, e.g., Dijkstra's algorithm.
  Still, it can be observed that $DP(\vex, \alpha)$ is a directed acyclic graph, and moreover, finding the lightest path can be done in a layer-by-layer manner in time $\Oh(|V(DP(\vex, \alpha))|\cdot L_{\max}) = \Oh(nt L_{\max}^2)$; cf. \cite[Lemma 3.4]{HemmeckeEtAl2013}.
  The claimed run time follows from the bound on the maximum size of a~layer~\eqref{LmaxBound}.
\end{proof}

\subsection{Step lengths}
\label{sec:longSteps}
We have shown how to find a feasible step $\veh$ for any given step length $\alpha \in \N$ such that $f(\vex + \alpha \veh) \leq f(\vex + \alpha \veg)$ for any feasible step $\veg \in\cGEn$.
Now, we will show that there are not too many step lengths that need to be considered in order to find a Graver\hy{}best step which, by Proposition~\ref{prop:graverbest}, leads to a good bound on the total required number of steps.
This is the case in particular if we have an instance whose feasible solutions are all contained in a box \mbox{$\vel \leq \vex \leq \veu$} with \mbox{$\|\vel - \veu\|_\infty \leq N$}, as no steps of length $\alpha > N$ are feasible.

In the following, we use the proximity technique pioneered by Hochbaum and Shantikumar~\cite{HochbaumS:90} in the case of totally unimodular matrices and extended to the setting of Graver bases by Hemmecke, Köppe and Weismantel~\cite{HemmeckeEtAl2014}.
This technique allows to show that, provided some structure of the constraints (e.g., total unimodularity or bounded $\ell_\infty$-norm of its Graver elements), the continuous optimum is not too far from the integer optimum of the problem.

\begin{proposition}{(\cite[Theorem 3.14]{HemmeckeEtAl2014})}
\label{thm:proximity}
  Consider a combinatorial $n$\hy{}fold IP $(IP)_{E^{(n)}, \veb, \vel, \veu, f}$.
  Then for any optimal solution $\hat{\vex}$ of the~continuous relaxation of $(IP)_{E^{(n)}, \veb, \vel, \veu, f}$ there is an optimal solution $\vex^*$ of $(IP)_{E^{(n)}, \veb, \vel, \veu, f}$ with
  \[
    \|\hat{\vex} - \vex^*\|_\infty \leq nt \cdot \max \left\{\|\veg\|_\infty \mid \veg \in \cGEn\right\} .
  \] 
\end{proposition}

This allows us then to reduce the original instance to an equivalent instance contained in a small box, as hinted at above.

\begin{lemma}[equivalent bounded instance]
\label{lem:equiv_bounded_instance}
  Let $(IP)_{E^{(n)},\veb, \vel, \veu, f}$ be a combinatorial $n$-fold IP of size $L = \l\veb, \vel, \veu, f\r$.
  With one call to an optimization oracle of its continuous relaxation, one can construct $\hat{\vel}, \hat{\veu} \in \Z^{nt}$ such that
  \[
    \min \big\{f(\vex) \mid \A\vex = \veb, \vel \leq \vex \leq \veu \big\} = \min \big\{f(\vex) \mid \A\vex = \veb, \hat{\vel} \leq \vex \leq \hat{\veu} \big\},
  \]
  and $\| \hat{\veu} - \hat{\vel} \|_\infty \leq nt\cdot g(E)$.
\end{lemma}
\begin{proof}
  Returning to Proposition~\ref{thm:proximity}, observe that the~quantity $\max \left\{\|\veg\|_\infty \mid \veg \in \cGEn \right\}$ is bounded by $g(E)$ (Lemma~\ref{lem:GA_structure}).
  Hence, we can set new lower and upper bounds $\hat{\vel}$ and $\hat{\veu}$ defined by $\hat{l}_j^i := \max \left\{\lfloor \hat{x}_j^i \rfloor - n g(E), l_j^i \right\}$ and $\hat{u}_j^i := \min \left\{\lceil \hat{x}_j^i \rceil + n g(E), u_j^i \right\}$, and Theorem~\ref{thm:proximity} assures that the~integer optimum also lies within the~new bounds.
\end{proof}

\subsection{Finishing the proof}
\begin{proof}[Proof of Theorem~\ref{thm:combinatorialnfold}.]~
  We proceed in three steps.
  
  \smallskip
  \noindent
  \textit{Step 1: Bounding the feasible region.}
  First, we use Lemma~\ref{lem:equiv_bounded_instance} to construct new lower and upper bounds $\hat{\vel}, \hat{\veu}$ satisfying $\| \hat{\veu} - \hat{\vel} \|_\infty \leq nt\cdot g(E)$ and preserving the optimal value of $(IP)_{E^{(n)},\veb, \vel, \veu, f}$.
  Thus, we shall replace $\vel$ and $\veu$ by $\hat{\vel}$ and $\hat{\veu}$ from now on and assume that $\|\veu - \vel\|_\infty \leq nt\cdot g(E)$.
  
  \smallskip
  \noindent
  \textit{Step 2: Optimization.}
  Let us assume that we have an initial feasible solution~$\vex_0$ of $(IP)_{E^{(n)},\veb, \vel, \veu, f}$.
  Given a step length $\alpha \in \N$, it takes time $t^{\Oh(r)} (\Delta r)^{\Oh(r^2)} n$ by Lemma~\ref{lem:optimality_cert} to find a feasible step $\veh$ satisfying $f(\vex + \alpha \veh) \leq f(\vex + \alpha \veg)$ for all $\veg \in \cGEn$.
  Recall that no $\veg \in \cGEn$ can be feasible for $\alpha > n g(E)$ by our bound on $\| \veu - \vel \|_\infty$.
  Thus, applying Lemma~\ref{lem:optimality_cert} for all $\alpha \in [nt\cdot g(E)]$ and choosing a pair $(\alpha, \veh)$ which minimizes $f(\vex + \alpha \veh)$ surely finds a Graver\hy best step in time $t^{\Oh(r)} (\Delta r)^{\Oh(r^2)} n^2$.
  In order to reach the optimum, by Proposition~\ref{prop:graverbest} we need to make at most $(2nt - 2) \cdot \Oh(L)$ Graver\hy{}best steps, where $L = \l \veb, \mathbf{0}, \veu, f \r$; this is because $\Oh(L)$ is an upper bound on $f(\vex^0) - f(\vex^*)$ for some optimal solution~$\vex^*$ of $(IP)_{E^{(n)},\veb, \vel, \veu, f}$.
  In total, we need time $t^{\Oh(r)} (\Delta r)^{\Oh(r^2)} n^3 L$.
  
  \smallskip
  \noindent
  \textit{Step 3: Feasibility.}
  Now we are left with the task of finding a starting feasible solution $\vex_0$ in the case when we do not have it.
  We follow the lines of Hemmecke et al.~\cite[Lemma 3.8]{HemmeckeEtAl2013} and solve an auxiliary combinatorial $n$-fold IP given by the bimatrix $\bar{E} = \left(\begin{smallmatrix}\bar{D}\\\bar{A}\end{smallmatrix}\right)$ with $\bar{D} := (D ~ I_{r} ~ -I_{r} ~ \mathbf{0})$ and $\bar{A} := (A ~ \mathbf{1}_{2r+1}) = \mathbf{1}^{\transpose} \in \Z^{t+2r+1}$, where $I_r$ is the identity matrix of dimension~$r$, $\mathbf{0}$ is a column vector of length~$r$ and $\mathbf{1}_{2r+1}$ is the vector of all $1$s of length $2r+1$.
  The variables $\bar{\vex}$ of this auxiliary problem have a natural partition into~$nt$ variables~$\vex$ corresponding to the original problem fixed at the beginning of Sect.~\ref{sec:maintheorem}, and $n(2r+1)$ new auxiliary variables~$\tilde{\vex}$.
  Keep the original lower and upper bounds on~$\vex$ and introduce a lower bound $0$ and upper bound $ntg(E)\Delta$ on each auxiliary variable.
  Finally, let the new objective be $f(\vex) = \bar{\vew}^\transpose\bar{\vex}$, i.e., the sum of the auxiliary variables.
  Observe that it is easy to construct an initial feasible solution by setting $\vex = \vel$ and computing $\tilde{\vex}$ accordingly: $\tilde{\vex}$ serve the role of slack variables, and the slack in any constraint is at most $nt^2g(E)\Delta$ by the fact that $\|\veu - \vel\|_\infty \leq nt\cdot g(E)$ and $\Delta = 1 + \|D\|_\infty$.
    
  Then, applying the optimization algorithm described in the beginning of the proof either finds a solution to the auxiliary problem with objective value~$0$, implying $\tilde{\vex} = \mathbf{0}$, and thus $\vex$ is feasible for the original problem, or no such solution exists, meaning that the original problem is infeasible.
\end{proof}

\section{Applications}
\label{sec:applications}
Let us repeat the main result of this paper:
\begin{reptheorem}{thm:combinatorialnfold}[repeated]
  \mainresult
\end{reptheorem}

In applications, it is practical to use combinatorial $n$-fold IP formulations which contain inequalities.
Given an $n$-fold IP (in particular a combinatorial $n$-fold IP), we call the upper rows $(D~D~\cdots~D) \vex = \veb^{0}$ \emph{globally uniform constraints}, and the lower rows $A \vex^i = \veb^i$, for all $i \in [n]$, \emph{locally uniform constraints}.
So we first show that introducing inequalities into a combinatorial $n$-fold IP is possible.
However in the case of globally uniform constraints, we need a slightly different approach than in a standard $n$-fold IP to keep the rigid format of a combinatorial $n$-fold IP.

\paragraph{Inequalities in locally uniform constraints.}
We add $n$ variables $x_{t+1}^i$ for all $i \in [n]$ and we replace $D$ with $(D~\mathbf{0})$.
For each row $i$ where we wish to enforce $\mathbf{1}^{\transpose} \vex^i \leq b^i$, we set the upper bound on $u_{t+1}^i = b^i$ and lower bound $l_{t+1}^i = 0$.
Similarly, for each row $i$ where we wish to enforce $\mathbf{1}^{\transpose} \vex^i \geq b^i$, we set a lower bound $l_{t+1}^i = -b^i$ and an upper bound $u_{t+1}^i = 0$.
For all remaining rows we set $l_{t+1}^i = u_{t+1}^i = 0$, enforcing $\mathbf{1}^{\transpose} \vex^i = b^i$.

\paragraph{Inequalities in globally uniform constraints.}
We replace $D$ with $(D~I_r)$, where~$I_r$ is the $r \times r$ identity matrix.
Thus, we have introduced $r$ new variables $x_{t+j}^i$ with $i \in [n]$ and $j \in [r]$; however, we enforce them all to be $0$ by setting $l_{t+j} = u_{t+j}^i = 0$ for all $i \in [n]$ and $j \in [r]$.
Next, we introduce an $(n+1)$-st brick, set $u_j^{n+1} = 0$ for all $j \in [t]$ and set $b^{n+1} = \|D\|_\infty \cdot \|(b^1, \dots, b^n)\|_1$.
Then, for each row $i \in [r]$ where we wish to enforce a ``$\leq$'' inequality, we set $l_{t+i}^{n+1} = 0$ and $u_{t+i}^{n+1} = \|\veb\|_\infty$, and for each row $i \in [r]$ with a ``$\geq$'' inequality, we set $l_{t+i}^{n+1} = -\|\veb\|_\infty$ and $u_{t+i}^{n+1} = 0$.
We let $l_{t+i}^{n+1} = u_{t+r+i}^{n+1} = 0$ for equality.
To enforce a strict inequality ``$<$'', we proceed as for ``$\leq$'' and increase the corresponding right hand side of the inequality by one; similarly for ``$>$''.

\subsection{Weighted Set Multicover}
\label{subsec:wsm}

We demonstrate Theorem~\ref{thm:combinatorialnfold} on the following problem:
\prob{\textsc{Weighted Set Multicover}}
{A~universe $U$ of size $k$, a set system represented by a multiset \mbox{$\F = \{F_1, \dots, F_n\} \subseteq 2^U$}, weights $w_1, \dots, w_n \in \N$, demands $d_1, \dots, d_k \in \N$.}
{A~multisubset $\F' \subseteq \F$ minimizing $\sum_{F_i \in \F'} w_i$ and satisfying \mbox{$\big|\{i \mid F_i \in \F', j \in F_i\}\big| \geq d_j$} for all $j \in [k]$.}

We a ready to prove Theorem~\ref{thm:weightedSetMulticover}:
\begin{proof}[Proof of Theorem~\ref{thm:weightedSetMulticover}.]
  Observe that there are at most $2^k$ different sets \mbox{$F \in 2^U$}; we classify each pair $(F, w)$ in the input into one of $T \leq 2^k$ different types.
  Moreover, for any two pairs $(F, w)$ and $(F, w')$ with $w \leq w'$, for any solution containing $(F', w')$ and not containing $(F, w)$ there is another solution which is at least as good and contains $(F, w)$.
  We thus order the pairs $(F,w),(F,w'),\hdots$ by non-decreasing weight, so that lighter elements will be used before heavier ones in any optimal solution.

  This allows us to represent the input instance in a succinct way by $T$ functions $g^i, \dots, g^T: [n] \to \N$ such that, for any $i \in [T]$, $g^i(k)$ is defined as the sum of the $k$ lightest elements of type $i$ or $+\infty$ in case that there are less than $k$ elements of type $i$.
  Observe that since each $g^i$ is a partial sum of a non-decreasing sequence of weights, it is a convex function.

  We construct a combinatorial $n$-fold IP to solve the problem.
  Let $x_{\vef}^\tau$ for each $\vef \in 2^U$ and each $\tau \in [T]$ be a variable.
  Let $l_{\vef}^\tau = u_{\vef}^\tau = 0$ for each $\vef \in 2^U$ such that $\vef \neq F^\tau$, and let $l_{\vef}^\tau = 0$ and $u_{\vef}^\tau = n$ for $\vef = F^\tau$.
  The variable $x_{\vef}^\tau$ with $\vef = F^\tau$ represents the number of sets of type $\tau$ in the solution.
  The IP formulation then reads
  \begin{align*}
    \min~& \sum_{\tau=1}^{T} g^{\tau}(x_{\vef}^\tau)
    & \mbox{s.t.}~\sum_{\tau=1}^{T} \sum_{\vef \in 2^U} f_i x_{\vef}^\tau \geq d_i, \qquad \mbox{for all } i \in [k] \\
    ~&& \sum_{\vef \in 2^U} x_{\vef}^\tau \leq n \qquad \mbox{for all } \tau \in [T];
  \end{align*}
  note that $f_i$ is $1$ if $i \in \vef$ and $0$ otherwise.
  Let us determine the parameters $\hat{\Delta},\hat{r},\hat{t},\hat{n}$ and $\hat{L}$ of this combinatorial $n$-fold IP instance.
  Clearly, the largest coefficient $||\hat{D}||_\infty$ is $1$, the number of globally uniform constraints $\hat{r}$ is~$k$, the number of variables per brick $\hat{t}$ is $2^k$, the number of bricks $\hat{n}$ is $T$, and the length of the input $\hat{L}$ is at most $\log n + \log w_{\max}$.
\end{proof}

\subsection{Stringology}
\label{subsec:strings}
To show that Theorem~\ref{thm:combinatorialnfold} can be used to speed up many previous result, we show a single-exponential algorithm for an artificial ``meta-problem'' called $\delta$-{\sc Multi Strings} which generalizes many previously studied problems:

\prob{$\delta$-\textsc{Multi Strings}}
{A~set of strings $S=\{s_1, \dots, s_k\}$, each of length $L$ over alphabet $\Sigma \cup \{\star\}$, distance lower and upper bounds \mbox{$d_1, \dots, d_k \in \N$} and $D_1, \dots, D_k \in \N$, distance function \mbox{$\delta: \Sigma^{*} \times \Sigma^{*} \to \N$} and a binary parameter $b \in \{0,1\}$.}
{An output string $y \in \Sigma^L$ with $d_i \leq \delta(y, s_i) \leq D_i$ for each $s_i \in S$, which minimizes $b \cdot \big(\sum_{i=1}^k \delta(y, s_i) \big)$.}
We call a distance function $\delta: \Sigma^* \times \Sigma^* \to \N$ \emph{character-wise wildcard-compatible} if $\delta(x,y) = \sum_{i=1}^L \delta(x[i], y[i])$ for any two strings $x,y \in \Sigma^L$, and $\delta(e, \star) = 0$ for all $e \in \Sigma$.

\begin{theorem}
\label{thm:multistrings}
  There is an algorithm that solves instances of $\delta$-{\sc Multi Strings} in time $K^{\Oh(k^2)} \Oh(\log L)$, where \mbox{$K = \max\left\{|\Sigma|, k, \max_{e,f \in \Sigma} \delta(e,f) \right\}$} and $\delta$ is any character-wise wildcard-compatible function.
\end{theorem}


When $\delta$ is the Hamming distance $d_H$, it is standard to first ``normalize'' the input to an equivalent instance over the alphabet $[k]$~\cite[Lemma 1]{HermelinRozenberg2015}.
Thus, for $\delta = d_H$ we get rid of the dependence on $|\Sigma|$:
\begin{theorem}
\label{thm:multistrings_hamming}
  $d_H$-{\sc Multi Strings} can be solved in time $k^{\Oh(k^2)} \Oh(\log L)$.
\end{theorem}
Then, Theorem~\ref{thm:strings_singleexp} below is a simple corollary of Theorem~\ref{thm:multistrings_hamming} and the fact that $\delta$-\textsc{Multi Strings} generalizes all the listed problems; see problem definitions and Table~\ref{tab:strings_results} below.

%
%

\begin{proof}[Proof of Theorem~\ref{thm:multistrings}]
  Let us fix an instance of $\delta$-{\sc Multi Strings}.
  We create an instance of combinatorial $n$-fold IP and show how solving it corresponds to solving the original $\delta$-\textsc{Multi Strings} problem.
    
  As is standard, we represent the input as an $L \times k$ matrix $C$ with entries from $\Sigma \cup \{\star\}$ whose rows are the input strings $s_1, \dots, s_k$.
  There are at most $T = (|\Sigma|+1)^k$ different \emph{input column types}; let $n_{\vece}$ be the number of columns of type $\vece \in (\Sigma \cup \{\star\})^k$ and denote $\T_c \subseteq (\Sigma \cup \{\star\})^k$ the set of input column types.
  A~solution can be represented as an $L \times (k+1)$ matrix with entries from $\Sigma \cup \{\star\}$ whose last row does not contain any $\star$ symbol.
  Thus, there are at most $(|\Sigma|+1)^k \cdot |\Sigma|$ \emph{solution column types} $\vealpha = (\vece, f) \in \big((\Sigma \cup \{\star\})^k \times \Sigma\big)$ and we denote $\T_s = \big((\Sigma \cup \{\star\})^k \times \Sigma\big)$ the set of all solution column types.
  We say that an input column type $\vece \in \T_c$ is \emph{compatible} with a solution column type $\vealpha \in \T_s$ if $\vealpha = (\vece, f)$ for some $f \in \Sigma$.
    
  Let us describe the combinatorial $n$-fold IP formulation.
  It consists of variables $x_{\vealpha}^{\vece}$ for each $\vealpha \in \T_s$ and each $\vece \in \T_c$.
  Intuitively, the variable $x_{\vealpha}^{\vece}$ encodes the number of columns $\vealpha$ in the solution; however, to obey the format of combinatorial $n$-fold IP, we need a copy of this variable for each brick, hence the upper index $\vece$.
  We set an upper bound $u_{\vealpha}^{\vece} = \|\veb\|_\infty$ for each two compatible~$\vece$ and $\vealpha$, and we set $u_{\vealpha}^{\vece} = 0$ for each pair which is not compatible; all lower bounds are set to $0$.
  The locally uniform constraints are simply $\sum_{\vealpha \in \T_s} x_{\vealpha}^{\vece} = n_{\vece}$ for all $\vece \in \T_c$.
  The globally uniform constraints are
    \begin{align*}
      \sum_{\vece \in \T_c} \  \sum_{\vealpha = (\vece', f) \in \T_s} \delta(f, e'_i) x_{\vealpha}^{\vece}& \geq d_i & \mbox{for all } s_i \in S \\
      \sum_{\vece \in \T_c} \ \sum_{\vealpha = (\vece', f) \in \T_s} \delta(f, e'_i) x_{\vealpha}^{\vece} & \leq D_i & \mbox{for all } s_i \in S
    \end{align*}
  and the objective is
  \[
    \min b \cdot \left(\sum_{i=1}^k \ \sum_{\vece \in \T_c} \ \sum_{\vealpha = (\vece', f) \in \T_s} \delta(f, e'_i) x_{\vealpha}^{\vece}\right) \enspace .
  \]
  We then apply Theorem~\ref{thm:combinatorialnfold} with the following set of parameters:
  \begin{itemize}
     \item $\hat{\Delta}$ is one plus the largest coefficient in $D$, which is $\displaystyle{1+\max_{e,f \in \Sigma} \delta(e,f) \leq 1+ K}$,
     \item $\hat{r}$ is the number of globally uniform constraints, which is $2k$,
     \item $\hat{t}$ is the number of variables per brick, which is \mbox{$|\T_s| \leq (|\Sigma|+1)^k |\Sigma| \leq K^{2k}$},
     \item $\hat{n}$ is the number of bricks, which is $|\T_c| \leq (|\Sigma|+1)^k \leq K^k$, and,
      \item $\hat{L}$ is the size of the input $\l \veb, \mathbf{0}, \veu, \vew \r \leq \log L$. \qedhere
  \end{itemize}
\end{proof}

The problem definitions follow.
Some problems reduce to solving polynomially (in $L$) or $k^k$ many instances of \textsc{Closest String}.
In such a case, we mention the fact after introducing the problem, and say that the problem \emph{poly-reduces} or \emph{\FPT-reduces} to {\sc Closest String}.

\prob{{\sc Closest String}}
{Strings $s_1, \dots, s_k \in \Sigma^L$, $d \in \N$.}
{A~string $y \in \Sigma^L$ such that $d_H(y, s_i) \leq d$ for all $s_i \in S$.}

\prob{{\sc Farthest String}}
{Strings $s_1, \dots, s_k \in \Sigma^L$, $d \in \N$.}
{A~string $y \in \Sigma^L$ such that $d_H(y, s_i) \geq d$ for all $s_i \in S$.}

\prob{{\sc $d$-Mismatch}}
{Strings $s_1, \dots, s_k \in \Sigma^L$, $d \in \N$.}
{A~string $y \in \Sigma^{L'}$ with $L' \leq L$ and a position $p \in [L-L']$ such that $d_H(y, s_{i,p,L'}) \geq d$ for all $s_i \in S$, where $s_{i,p,L'}$ is the substring of $s_i$ of length $L'$ starting at position $p$.}
\textbf{Note.} Gramm et al.~\cite{GrammEtAl2003} observe that $d$-{\sc Mismatch} poly-time reduces to {\sc Closest String}.

\prob{{\sc Distinguishing String Selection} (DSS)}
{Bad strings $S=\{s_1, \dots, s_{k_1}\}$, good strings $S' = \{s'_1, \dots, s'_{k_2}\}$, $k=k_1 + k_2$, $d_1, d_2 \in \N$, all strings of length $L$ over alphabet $\Sigma$.}
{A~string $y \in \Sigma^L$ such that $d_H(y, s_i) \leq d_1$ for each bad string $s_i$ and $d_H(y, s'_i) \geq L-d_2$ for each good string $s'_i$.}

\prob{{\sc Neighbor String}}
{Strings $s_1, \dots, s_k \in \Sigma^L$, $d_1, \dots, d_k \in \N$.}
{A~string $y \in \Sigma^L$ such that $d_H(y, s_i) \leq d_i$ for all $i \in [k]$.}
\textbf{Note.} \textsc{Neighbor String} is studied by Nishimura and Simjour~\cite{NishimuraSimjour2012}.
It generalizes DSS: given an instance of DSS, create an instance of {\sc Neighbor String} with $d_i = 0$ and $D_i = d_1$ for all bad strings, and $d_i = d_2$ and $D_i = L$ for all good strings.

\prob{{\sc Closest String with Wildcards}}
{Strings $s_1, \dots, s_k \in (\Sigma \cup \{\star\})^L$, $d \in \N$.}
{A~string $y \in \Sigma^L$ such that $d_H(y, s_i) \leq d$ for all $i \in [k]$, where $d_H(e, \star) = 0$ for any $e \in \Sigma$.}

\prob{{\sc Closest to Most Strings} (also known as {\sc Closest String with Outliers})}
{Strings $s_1, \dots, s_k \in \Sigma^L$, $o \in [k]$, $d \in \N$.}
{A~string $y \in \Sigma^L$ and a set of outliers $O \subseteq \{s_1, \dots, s_k\}$ such that $d_H(y, s_i) \leq d$ for all $s_i \not\in O$ and $|O| \leq o$.}
\textbf{Note.} {\sc Closest to Most Strings} is \FPT-reducible to {\sc Closest String} with parameter $k$~\cite{BoucherMa2011}.

\prob{$c$-{\sc Hamming Radius Clustering} ($c$-HRC)}
{Strings $s_1, \dots, s_k \in \Sigma^L$, $d \in \N$.}
{A~partition of $\{s_1, \dots, s_k\}$ into $S_1, \dots, S_c$ and output strings $y_1, \dots, y_c \in \Sigma^L$ such that $d(y_i, s_j) \leq d$ for all $i \in [c]$ and $s_j \in S_i$.}


\prob{{\sc Optimal Consensus}}
{Strings $s_1, \dots, s_k \in \Sigma^L$, $d \in \N$.}
{A~string $y \in \Sigma^L$ such that $d_H(y, s_i) \leq d$ for all $i \in [k]$ and $\sum_{i \in [k]} d_H(y, s_i)$ is minimal.}

See Table~\ref{tab:strings_results} for a summary of our improvements for the above-mentioned problems.

\begin{table}[htb]
  \centering
  \resizebox{\textwidth}{!}{%
    \begin{tabular}{lll}
      \toprule
      Problem   & Specialization of $\delta$-{\sc Multi Strings} & Previous best run time / hardness\\
      \midrule
      {\sc Closest String} & $D_i = d$ for all $i \in [k]$ & $2^{2^{\Oh(k \log k)}} \log^{\Oh(1)} n$ \cite{GrammEtAl2003}\\
      \midrule
      {\sc Farthest String} & $d_i = d$ for all $i \in [k]$ & $2^{2^{\Oh(k \log k)}} \log^{\Oh(1)} n$ \cite[``implicit'']{GrammEtAl2003}\\
      \midrule
      $d$-{\sc Mismatch} & poly-reduces to {\sc Closest String}~\cite{GrammEtAl2003} & $2^{2^{\Oh(k \log k)}} \log^{\Oh(1)} n$~\cite{GrammEtAl2003}\\
      \midrule
      {\sc Distinguishing String Selection} & special case of {\sc Neighbor String} & $2^{2^{\Oh(k \log k)}} \log^{\Oh(1)} n$~\cite[``implicit'']{GrammEtAl2003}\\
      \midrule
      {\sc Neighbor String} & $D_i = d_i$ for all $i \in [k]$ & $2^{2^{\Oh(k \log k)}} \log^{\Oh(1)} n$~\cite[``implicit'']{GrammEtAl2003}\\
      \midrule
      {\sc Closest String with Wildcards} & $D_i = d$ for all $i \in [k]$ & $2^{2^{\Oh(k \log k)}} \log^{\Oh(1)} n$ \cite{HermelinRozenberg2015}\\
      \midrule
      {\sc Closest to Most Strings} & \FPT-reduces to \textsc{Closest String}~\cite{BoucherMa2011} & $2^{2^{\Oh(k \log k)}} \log^{\Oh(1)} n$~\cite[implicit]{GrammEtAl2003} \\
      \midrule
      $c$-HRC & \FPT-reduces to \textsc{Closest String}~\cite{AmirEtAl2014} & $2^{2^{\Oh(k \log k)}} \log^{\Oh(1)} n$~\cite[implicit]{GrammEtAl2003} \\
      \midrule
      {\sc Optimal Consensus} & $D_i=d$ for all $i \in [k]$, $b=1$ & \FPT for $k=3$, open for $k>3$~\cite{AmirEtAl2011} \\
      \bottomrule
  \end{tabular}}
  \caption{If the ``specialization'' row does not contain a value, it means its ``default'' value is assumed.
    The default values are $\delta = d_H$, $b=0$, and for all $i \in [k]$, $d_i = 0$, $D_i = L$.
    In each row corresponding to a problem, the last column gives the run time of the algorithm with the slowest-growing dependency on $k$.
    Most problems either reduce to {\sc Closest String} and thus derive their time complexity from the result of Gramm et al. even though the original paper does not mention these problems; in that case, we write~\cite[implicit]{GrammEtAl2003}.
    For some problems, no fixed-parameter algorithm was known before, but it is not difficult to see that the ILP formulation of Gramm et al. could be modified to model these problems as well; in that case, we write~\cite[``implicit'']{GrammEtAl2003}.}
\label{tab:strings_results}
\end{table}

\subsection{Computational social choice}
\label{subsec:bribery}
For simplicity, we only show how Theorem~\ref{thm:combinatorialnfold} can be applied to speed up the $\mathcal{R}$-{\sc Swap Bribery} for two representative voting rules $\mathcal{R}$.
Let us first introduce the necessary definitions and terminology.

\medskip
\noindent
\textbf{Elections.}
An election~$(C,V)$ consists of a set $C$ of candidates and a set~$V$ of voters, who indicate their preferences over the candidates in $C$, represented via a \emph{preference order} $\pref_v$ which is a total order over $C$.
For ranked candidates~$c$ we denote by $\textrm{rank}(c,v)$ their rank in~$\pref_v$; then $v$'s most preferred candidate has rank 1 and their least preferred candidate has rank $|C|$.
For distinct candidates~$c,c'\in C$, we write $c\pref_v c'$ if voter~$v$ prefers~$c$ over~$c'$.
To simplify notation, we sometimes identify the candidate set~$C$ with $\{1,\hdots,|C|\}$, in particular when expressing permutations over $C$.
We sometimes identify a voter~$v$ with their preference order $\pref_v$, as long as no confusion arises.

\medskip
\noindent
\textbf{Swaps.}
Let $(C,V)$ be an election and let $\pref_v\in V$ be a voter.
For candidates $c,c'\in C$, a \emph{swap} $s = (c,c')_v$ means to exchange the positions of $c$ and $c'$ in~$\pref_v$; denote the perturbed order by~$\pref_v^s$.
A~swap~$(c,c')_v$ is \emph{admissible in $\pref_v$} if $\rank(c,v) = \rank(c',v)-1$.
A~set $S$ of swaps is \emph{admissible in $\pref_v$} if they can be applied sequentially in~$\pref_v$, one after the other, in some order, such that each one of them is admissible.
Note that the perturbed vote, denoted by $\pref_v^S$, is independent from the order in which the swaps of $S$ are applied.
We also extend this notation for applying swaps in several votes and denote it $V^S$.
We specify $v$'s cost of swaps by a function $\sigma^v: C\times C\rightarrow \mathbb{Z}$.

\medskip
\noindent
\textbf{Voting rules.}
A~voting rule~$\mathcal R$ is a function that maps an election $(C,V)$ to a subset $W\subseteq C$ of \emph{winners}.
Let us define two significant classes of voting rules:

\textit{Scoring protocols.}
A~scoring protocol is defined through a vector $\ves = (s_1,\hdots,s_{|C|})$ of integers with \mbox{$s_1\geq \cdots\geq s_{|C|} \geq 0$}.
For each position ${p\in\{1,\hdots,|C|\}}$, value $s_p$ specifies the number of points that each candidate $c$ receives from each voter that ranks $c$ as $j^{\text{th}}$ best.
Any candidate with the maximum number of points is a winner.
Examples of scoring protocols include the Plurality rule with $\ves = (1,0,\ldots,0)$, the $d$-Approval rule with $\ves = (1,\ldots,1,0,\ldots,0)$ with $d$ ones, and the Borda rule with $\ves = (|C|-1, |C|-2, \ldots, 1, 0)$.
Throughout, we consider only \emph{natural} scoring protocols for which $s_1 \leq |C|$; this is the case for the aforementioned popular rules.

\textit{C1 rules.}
A~candidate~$c\in C$ is a \emph{Condorcet winner} if any other~$c'\in C \setminus \{c\}$ satisfies $|\{\pref_v \in V \mid c\pref_v c' \}| > |\{v \in V \mid c' \pref_v c\}|$.
A~voting rule is \emph{Condorcet-consistent} if it selects the Condorcet winner in case there is one.
Fishburn~\cite{Fishburn1977} classified all Condorcet-consistent rules as C1, C2 or C3, depending on the kind of information needed to determine the winner.
For candidates $c,c' \in C$ let $v(c,c')$ be the number of voters who prefer $c$ over $c'$, that is, $v(c,c') = |\{\pref_v \in V \mid c \pref_v c'\}|$;
we write $c <_M c'$ if $c$ beats $c'$ in a head-to-head contest, that is, if $v(c,c') > v(c',c)$.

A~rule \emph{$\mathcal R$ is C1} if knowing $<_M$ suffices to determine the winner, that is, for each pair of candidates $c,c'$ we know whether $v(c,c') > v(c',c), v(c,c') < v(c',c)$ or $v(c,c') = v(c',c)$.
An example is the Copeland$^\alpha$ rule for $\alpha \in [0,1]$, which specifies that for each head-to-head contest between two distinct candidates, if some candidate is preferred by a majority of voters then they obtain one point and the other candidate obtains zero points, and if a tie occurs then both candidates obtain $\alpha$ points; the candidate with largest sum of points wins.

\prob{$\R$-{\sc Swap Bribery}}
{An election $(C,V)$, a designated candidate $c^{\star} \in C$ and swap costs $\sigma^v$ for $v \in V$.}
{A~set $S$ of admissible swaps of minimum cost so that $c^{\star}$ wins the election $(C, V^S)$ under the rule $\R$.}

We say that two voters $v,v'$ are of the same \emph{type} if $\pref_{v} = \pref_{v'}$ and $\sigma^{v} = \sigma^{v'}$.
%
We are ready to prove Theorem~\ref{thm:applicationBribery}:
\begin{proof}[Proof of Theorem~\ref{thm:applicationBribery}.]
  Let $n_1, \dots, n_{T}$ be the numbers of voters of given types.
  Let $x_j^i$ for $j \in [|C|!]$ and $i \in [T]$ be a variable encoding the number of voters of type $i$ that are bribed to be of order $j$ in the solution.
  With slight abuse of notation, we denote $\sigma^i(i,j)$ the cost of bribery for a voter of type $i$ to change order to $j$ (as by \cite[Proposition~3.2]{ElkindEtAl2009} this cost is fixed).
  Regardless of the voting rule $\R$, the objective and the locally uniform constraints are identical:
\[
    \min \sum_{i=1}^{T} \sum_{j=1}^{|C|!} \sigma^i(i,j) x_j^i~\mbox{subject to}~\sum_{j=1}^{|C|!} x_j^i = n_i~\mbox{for all}~i \in [T] \enspace .
\]
  The number of variables per brick $\hat{t}$ is $|C|!$, the number of bricks $\hat{n}$ is $T$, and the size of the instance $\hat{L}$ is $\log n + \log (|C|^2 \sigma_{\max})$, because at most $|C|^2$ swaps suffice to permute any order $i \in \left[|C|!\right]$ to any other order $j \in [|C|!]$~\cite[Proposition~3.2]{ElkindEtAl2009}.
  Let us now describe the globally uniform constraints separately for the two classes of voting rules which we study here.
    
  \paragraph{Natural scoring protocol.}
  Let $\ves = (s_1, \dots, s_{|C|})$ be a natural scoring protocol, i.e., \mbox{$s_1 \geq \dots \geq s_{|C|}$} and \mbox{$\|\ves\|_\infty \leq |C|$}.
  With slight abuse of notation, we denote $s_j(c)$, for $j \in [|C|!]$ and $c \in C$, the number of points obtained by candidate $c$ from a voter of order $j$.
  The globally uniform constraints then enforce that $c^{\star}$ gets at least as many points as any other candidate~$c$:
  \begin{align*}
    \sum_{i=1}^{T} \sum_{j=1}^{|C|!} s_j(c) x_j^i &\leq \sum_{i=1}^{T} \sum_{j=1}^{|C|!} s_j(c^{\star}) x_j^i & \mbox{for all } c \in C, c \neq c^{\star} \enspace .
  \end{align*}
  The number $\hat{r}$ of these constraints is $|C|-1$, and the largest coefficient in them is \mbox{$\|\ves\|_\infty \leq |C|$} (as is $\ves$ being a natural scoring protocol); therefore, $\hat{\Delta} = 1 + \|\ves\| \leq 1 + |C|$.
    
  \paragraph{Any C1 rule.}
  Let $\alpha_j(c,c')$ be $1$ if a voter with order $j \in [|C|!]$ prefers $c$ to $c'$ and $0$ otherwise.
  Recall that a voting rule is C1 if, to determine the winner, it is sufficient to know, for each pair of candidates $c,c'$, whether $v(c,c') > v(c',c), v(c,c') < v(c',c)$ or $v(c,c') = v(c',c)$; where $v(c,c') = |\{v \mid c \pref_v c'\}|$.
  We call a tuple $<_M \in \{<, =, >\}^{|C|^2}$ a \emph{scenario}.
  Thus, a C1 rule can be viewed as partitioning the set of all scenarios into those that select $c^{\star}$ as a winner and those that do not.
  Then, it suffices to enumerate all the at most $3^{|C|^2}$ scenarios~$<_M$ where $c^{\star}$ wins, and for each of them to solve a combinatorial $n$-fold IP with globally uniform constraints enforcing the scenario $<_M$:
  \begin{align*}
    \sum_{i=1}^{T} \sum_{j=1}^{|C|!} \alpha_j(c,c') x_j^i &> \sum_{i=1}^{T} \sum_{j=1}^{|C|!} \alpha_j(c',c) x_j^i & \mbox{for all } c, c' \in C \mbox{ s.t. } c <_M c' \\
    \sum_{i=1}^{T} \sum_{j=1}^{|C|!} \alpha_j(c,c') x_j^i &= \sum_{i=1}^{T} \sum_{j=1}^{|C|!} \alpha_j(c',c) x_j^i & \mbox{for all } c, c' \in C  \mbox{ which are incomparable.}
  \end{align*}
  The number $\hat{r}$ of these constraints is $\binom{|C|}{2} \leq |C|^2$, and the largest coefficient in them is 1, so $\hat{\Delta} = 2$.
  The proof is finished by plugging in the values $\hat{\Delta}, \hat{r}, \hat{t}, \hat{n}$ and $\hat{L}$ into Theorem~\ref{thm:combinatorialnfold}.
\end{proof}

\subsection{Huge $n$-fold integer programming with small domains}
\label{subsec:hugenfold}
Onn introduces the {\sc Huge $n$-fold integer programming} problem~\cite{Onn2014,OnnSarrabezolles2015}, which concerns problems that can be formulated as an $n$-fold IP with the number of bricks $n$ given in binary.
Bricks are thus represented not explicitly, but succinctly by their multiplicity.
It is at first unclear if this problem admits an optimal solution which can be encoded in polynomial space, but this is possible by a theorem of Eisenbrand and Shmonin~\cite[Theorem~2]{EisenbrandShmonin2006}, as pointed out by Onn~\cite[Theorem~1.3~(1)]{Onn2014}.
Thus, the problem is as follows:

Let $E = \left(\begin{smallmatrix}D\\A\end{smallmatrix}\right) \in \Z^{(r+s) \times t}$ be a bimatrix.
Let $T$ be a positive integer representing the \emph{number of types of bricks}.
We are given $T$ positive integers $n_1, \dots, n_{T}$ with $n = \sum_{i=1}^{T} n_i$ and vectors $\veb^0 \in \Z^r$ and $\vel^i, \veu^i \in \Z^t$ and $\veb^i \in \Z^s$, and a separable convex function $f^i$, for every $i\in [T]$.
For $i \in [T]$ and $\ell \in [n_i]$ we define the \emph{index function} $\iota$ as $\iota(i, \ell) := \left(\sum_{j=1}^{i-1} n_j \right) + \ell$.

We call an $n$-fold IP instance given by the constraint matrix $E^{(n)}$ with the right hand side $\hat{\veb}$ defined by $\hat{\veb}^{\iota(i,\ell)} := \veb^i$ and $\hat{\veb}^0 := \veb^0$, lower and upper bounds defined by $\hat{\vel}^{\iota(i, \ell)} := \vel^i$ and $\hat{\veu}^{\iota(i,\ell)} := \veu^i$ with the objective function $f(\vex) := \sum_{i=1}^n \sum_{\ell=1}^{m_i} f^i\big(\vex^{\iota(i,\ell)}\big)$ the \emph{huge instance}.

For $i \in [T]$ and $\ell \in [n_i]$, and a feasible solution $\vex$ of the huge instance, we say that the brick $\vex^{\iota(i, \ell)}$ is of \emph{type} $i$, and we say that~$\vex^{\iota(i,\ell)}$ has \emph{configuration} $\vecc \in \Z^t$ with $\hat{\vel}^i\le\vecc\le\hat{\veu}^i$ if $\vex^{\iota(i,\ell)} = \vecc$.
The \emph{succinct representation of $\vex$} is the set of tuples 
$\left\{ \left(\vex^{i,j}, m^{i,j} \right) \mid \mbox{$\vex$ has $m^{i,j}$ bricks of type $i$ with configuration $\vex^{i,j}$}\right\}$.

\medskip
\prob{\textsc{Huge $n$\hy fold Integer Programming}}
{
    Bimatrix $E = \left(\begin{smallmatrix}D\\A\end{smallmatrix}\right) \in \Z^{(r+s) \times t}$, positive integers $n_1, \dots, n_{T}$, $\veb^0 \in \Z^r$, for every $i \in [T]$ vectors $\vel^i, \veu^i \in \Z^t$ and $\veb^i \in \Z^s$ and a separable convex function~$f^i$.
}
{
    The succinct representation of an optimal solution, if such exists.
}

In the special case of small domains we obtain the following:
\begin{theorem}\label{thm:hugeNFold}
  Let $d_1, \dots, d_t \in \N$ be such that $d_j = \max_{i \in [n]} u_j^i - l_j^i$, $d_{\max} = \max_{j \in [t]} d_j$ and let $\delta = \prod_{j=1}^t d_j$.
  Then the huge $n$-fold IP problem can be solved in time $\delta^{\Oh(r)} (t d_{\max} \|D\|_\infty r)^{\Oh(r^2)} \Oh(T^3 \log n)$.\footnote{In fact, our result holds even in the case when $f$ is an arbitrary (i.e. non-convex) function, but this does not imply any more power because of bounded domains.}
\end{theorem}

This result is useful for the following reason.
Knop et al.~\cite{KnopEtAl2017} obtain $n$-fold IP formulations with small domains for the very general $\R$-\textsc{Multi Bribery} problem.
Together with Theorem~\ref{thm:hugeNFold}, this immediatelly implies an exponential speedup in the number of bricks, without having to reformulate these problems as combinatorial $n$-fold IPs.
However, there are still benefits in using Theorem~\ref{thm:combinatorialnfold} directly (as shown in previous sections) as it leads to better dependence on the respective parameters.

\begin{proof}[Proof of Theorem~\ref{thm:hugeNFold}]
  Let $E, n_1, \dots, n_{T}, \veb, \vel, \veu, f$ be an instance $I$ of \textsc{Huge $n$\hy{}fold IP}.
  First, we shall prove that we can restrict our attention to the case where $\vel = \mathbf{0}$ and $\veu^i \leq (d_1, \dots, d_t) = \ved$ for all $i \in [T]$.
  Consider a variable $x_j^i$ with $l_j^i \neq 0$ and any row $\vece^\transpose \vex = b$ of of the system $E^{(n)} \vex = \veb$.
  Because the contribution of $x_j^i$ to the right hand side is $e_j^i x_j^i$, we have that 
 \[
    \vece \vex = b \Leftrightarrow \vece \vex - e_j^i l_j^i = b - e_j^i l_j^i \enspace .
\]
  Let $I'$ be an instance of \textsc{Huge $n$-fold IP} obtained from $I$ by, for every row $\vece \vex = b$ and every variable $x_j^i$, changing the right hand side from $b$ to $b - e_j^i l_j^i$, and setting $u_j^i := u_j^i - l_j^i$ and $l_j^i := 0$.
  Clearly there is a bijection between the feasible solutions of $I$ and $I'$ such that if $\vex-\vel$ is a feasible solution of $I'$, $\vex$ is a feasible solution of $I$, and thus minimizing $f(\vex - \vel)$ over $I'$ is equivalent to minimizing $f(\vex)$ over $I$.
  Thus, from now on assume that $I$ satisfies $\vel = \mathbf{0}$ and $\veu^i \leq \ved$ for all $i \in [T]$.
    
  Let $\C^i$ for $i \in [T]$ be the set of all possible configurations of a brick of type~$i$, defined as $\C^i = \left\{\vecc \in \Z^t \mid A \vecc = \veb^i, \mathbf{0} \leq \vecc \leq \veu^i \right\}$ and let $\C = \prod_{j=1}^t [0\,:\,d_j]$ be the set of all configurations.
  Clearly, $\C^i \subseteq \C$ for all $i \in [T]$, and $|\C| = \delta$.
  Let $C \in \Z^{t \times \delta}$ be a matrix whose columns are all configurations from $\C$.
    
  We shall give a combinatorial $n$-fold IP formulation solving the huge $n$-fold IP instance $I$.
  The formulation contains variables $y_{\vecc}^i$ for each $\vecc \in \C$ and each $i \in [T]$ encoding how many bricks of type $i$ have configuration $\vecc$ in the solution of $I$. The formulation then is
    \begin{align}
             \min & & \hat{f}(\vey)                  & = \sum_{i=1}^{T} \sum_{\vecc \in \C} f^i(\vecc) y_{\vecc}^i & \label{eqn:hugenfold:obj}\\
      \mbox{s.t.} & & DC \vey                        & = \veb^0                                                    & \label{eqn:hugenfold:DC}\\
                  & & \mathbf{1}^{\transpose} \vey^i & = n_i                                                       & \mbox{for all } i \in [T] \label{eqn:hugenfold:local}\\
                  & & y_{\vecc}^i                    & = 0                                                         & \mbox{for all } \vecc \not\in \C^i, i \in [T] \label{eqn:hugenfold:notconf}\\
                  & & 0 \le y_{\vecc}^i              & \leq \|\veb\|_\infty                                        & \mbox{for all } \vecc \in \C^i, i \in [T] \enspace   .
    \label{eqn:hugenfold:conf}
    \end{align}
  It remains to verify that the formulation above corresponds to the huge $n$-fold IP instance $I$.
  The objective~\eqref{eqn:hugenfold:obj} clearly has the same value.
  Consider the globally uniform constraints~\eqref{eqn:hugenfold:DC}. In the huge $n$-fold IP instance, a configuration $\vecc \in \C^i$ of a brick of type $i \in [T]$ contributes $D\vecc$ to the right hand side in the first $r$ rows.
  This corresponds in our program to the column $D\vecc$ of the matrix~$DC$.
  The locally uniform constraints~\eqref{eqn:hugenfold:local} simply state that the solution needs to contain exactly $n_i$ bricks of type $i$.
  Finally, since a brick of type $i$ can never have a configuration $\vecc \not\in \C^i$ we set all variables $y_{\vecc}^i$ with $\vecc \not\in \C^i$ to zero with the upper bound~\eqref{eqn:hugenfold:notconf}, and place no restrictions on $y_{\vecc}^i$ with $\vecc \in \C^i$~\eqref{eqn:hugenfold:conf}.
    
  The parameters $\hat{\Delta},\hat{r},\hat{t},\hat{n},\hat{L}$ of the resulting combinatorial $n$-fold IP are:
  \begin{itemize}
    \item $\hat{\Delta}$ is one plus the largest coefficient in the upper matrix $\hat{D} = DC$, which is \mbox{$\|DC\|_\infty \leq td_{\max}\|D\|_\infty$},
    \item the number of globally uniform constraints $\hat{r} = r$,
    \item the number of variables in a brick $\hat{t} = \delta$,
    \item the number of bricks $\hat{n} = T$, and,
    \item the input length $\hat{L} = \l \hat{\veb}, \mathbf{0}, \hat{\veu}, \hat{f} \r \leq \log n \cdot \left(\max_{i \in [T]} \max_{\vecc \in \C^i} f^i(\vecc)\right)$. \qedhere
  \end{itemize}
\end{proof}

\subsection{Miscellaneous}
Recall from Sect.~\ref{sec:introduction} our comparison of the Gramm et al.~\cite{GrammEtAl2003} ILP for \textsc{Closest String} to a similar combinatorial $n$-fold IP:
\[
\begin{array}{ccccl|cccccl}
D_1    & D_2    & \cdots & D_{k^k} & \leq d  & \hspace{9.5em}~
& D & D & \cdots & D & \leq d
\\
\mathbf{1}^{\transpose}    & 0      & \cdots & 0 & = b^1 & \hspace{20em}   
& \mathbf{1}^{\transpose} & 0 & \cdots & 0 & = b^1 
\\
0      & \mathbf{1}^{\transpose}    & \cdots & 0 & = b^2 & \hspace{20em} 
& 0      & \mathbf{1}^{\transpose}    & \cdots & 0 & = b^2
\\
\vdots & \vdots & \ddots & \vdots &  =  \vdots & \hspace{20em}
& \vdots & \vdots & \ddots & \vdots &  =  \vdots
\\
0      & 0      & \cdots & \mathbf{1}^{\transpose} & = b^{k^k} & \hspace{20em}  
& 0      & 0      & \cdots & \mathbf{1}^{\transpose} & = b^{k^k},
\\
\end{array}
\]
where $D = (D_1~D_2~\dots~D_{k^k})$.

This similarity strongly suggests a general way how to construct the formulation on the right given the formulation on the left.
Since formulations like the one on the left are ubiquitous in the literature, this would immediatelly imply exponential speed-ups for all such problems.

\begin{definition}[Combinatorial pre-$n$-fold IP]
  Let $T, r, t_1, \dots, t_T \in \N$ and $D_i \in \Z^{r \times t_\tau}$ for each $\tau \in [T]$.
  Let $t = t_1 + \cdots t_T$ and $D = (D_1 \cdots D_T)$.
  Let
  \[
    F = \left[\begin{array}{cccc}
              D_1                     & D_2                        & \cdots & D_\tau \\
              \mathbf{1}^{\transpose} & 0                          & \cdots & 0      \\
              0                       & \mathbf{1}^{\transpose}    & \cdots & 0      \\
              \vdots                  & \vdots                     & \ddots & \vdots \\
              0                       & 0                          & \cdots & \mathbf{1}^{\transpose}\\
              \end{array}\right] \enspace .
  \]
  Moreover, let $\veb = (\veb^0, b^1, \dots, b^T)$ with $\veb^0 \in \Z^r$ and $b^\tau \in \Z$ for each $\tau \in [T]$, $\vel, \veu \in \Z^t$, and let $f: \Z^t \to \Z$ be a separable convex function.
  Then for $\diamondsuit \in \{<, \leq, =, \geq, >\}^{r+T}$, a \emph{combinatorial pre-$n$-fold IP} is the problem
  \[
    \min \big\{  f(\vex) \mid F \vex \diamondsuit \veb, \vel \leq \vex \leq \veu \big\} \enspace .
  \]
\end{definition}

\begin{corollary}
\label{cor:combinatorial_prenfold}
  Any combinatorial pre-$n$-fold IP with $L = \l \veb, \vel, \veu, f \r$ and $\Delta = 1 + \|D\|_\infty$ can be solved in time $t^{\Oh(r)} (\Delta r)^{\Oh(r^2)}\Oh(n^3 L) + \mathsf{oo}$, where $\mathsf{oo}$ is the time required by one call to an optimization oracle of the continuous relaxation of the given IP.
\end{corollary}
\begin{proof}
  We shall create a combinatorial $n$-fold IP instance based on the input combinatorial pre-$n$-fold IP.
  Let
  \[
    E^{(n)} = \left[\begin{array}{cccc}
                      D                       & D                          & \cdots & D \\
                      \mathbf{1}^{\transpose} & 0                          & \cdots & 0 \\
                      0                       & \mathbf{1}^{\transpose}    & \cdots & 0 \\
                      \vdots                  & \vdots                     & \ddots & \vdots \\
                      0                       & 0                          & \cdots & \mathbf{1}^{\transpose} \\
                    \end{array}\right] \enspace .
  \]
  Let $\bar{T}_i = \sum_{j=1}^\tau t_\tau$ for each $\tau \in [T]$.
  Let $\varphi: \bigcup_{\tau=1}^T \big(\{\tau \} \times [t_\tau]\big) \to [T] \times [t]$ be an injective mapping from the original variables to the new ones, defined as follows: for each $\tau \in [T]$ and $j \in [t_\tau]$, $\varphi(\tau,j) = (\tau, \bar{T}_{\tau-1} + j)$.
  We call any pair $(\tau,j) \in \N^2$ without a preimage in $\varphi$ \emph{dummy}.

  Then, for any $\tau \in [T]$ and $j \in [t_\tau]$, let $(\tau', j') = \varphi(\tau,j)$, and set $\hat{l}^{\tau'}_{j'} = l^\tau_j$, $\hat{u}^{\tau'}_{j'} = u^\tau_j$, and $\hat{f}^{\tau'}_{j'} = f^\tau_j$.
  For any $\tau \in [T]$ and $j \in [t]$ which form a dummy pair, set $\hat{u}^\tau_j = \hat{l}^\tau_j = 0$ and let $\hat{f}^\tau_j$ be the zero function.

  Now we see that \[\min \big\{ \hat{f}(\hat{\vex}) \mid E^{(n)} \hat{\vex} \diamondsuit \veb, \hat{\vel} \leq \hat{\vex} \leq \hat{\veu} \big\}\] is a combinatorial $n$-fold IP (with inequalities) and thus can be solved in the claimed time by Theorem~\ref{thm:combinatorialnfold}.
\end{proof}

\paragraph{Parameterizing by the number of numbers.}
Fellows et al.~\cite{FellowsGR:2012} argue that many problems' \NP-hardness construction is based on having many distinct objects, when it would be quite natural that the number of distinct objects, and thus the numbers representing them, is bounded by a parameter.
Fellows et al.~\cite[Theorem 5]{FellowsGR:2012} give a fixed-parameter algorithm for a certain Mealy automaton problem which models, for example the \textsc{Heat-sensitive scheduling} problem parameterized by the number of distinct heat levels.
Their algorithm relies on Lenstra's algorithm and has a doubly-exponential dependence on the parameter.
In the conclusions of their paper, the authors state that ``[o]ur main FPT result, Theorem 5, has a poor worst-case running-time guarantee. Can this be improved—at least in important special cases?''

It is easy to observe that their ILP only has $\sigma$ constraints, and that the coefficients are bounded by $|S|^{|S|^2}$.
Thus, applying our algorithm with $t = \Delta = \Oh(|S|^{|S|^2})$ and $r = \sigma$ yields an algorithm with runtime $|S|^{\Oh(\sigma |S|^2)}$, which is single-exponential in their (combined) parameter.

\paragraph{Scheduling meets fixed-parameter tractability.}
Mnich and Wiese~\cite{MnichWiese2015} study the parameterized complexity of fundamental scheduling problems, starting with \textsc{Makespan Minimization} on identical machines.
They give an algorithm with doubly-exponential dependence for parameter maximum job length~$p_{\max}$, and polylogarithmic dependence on the number $m$ of machines.
Knop and Kouteck{\'y}~\cite{KnopKoutecky2017} use standard $n$-fold IP to reduce the dependence on $p_{\max}$ to single-exponential, however their algorithm depends polynomially on~$m$.

We observe that in the proof of~\cite[Theorem 2]{MnichWiese2015} there is an ILP related to a set of configurations which is of size $p_{\max}^{p_{\max}}$.
The coefficients of this ILP are unbounded, but by~\cite[Lemma 1]{MnichWiese2015} we know that all the coefficients differ by at most $p_{\max}^{p_{\max}}$.
Moreover, because the number of machines $m$ is fixed before the construction of the ILP, we can appropriately subtract from the right hand sides and decrease all coefficients such that $\Delta \leq p_{\max}^{p_{\max}}$.
Then, take the constraints~\eqref{eqn:combnfoldip-cond1} as globally uniform and notice that there are $r = p_{\max}$ of them.
In conclusion, we have $t = \Delta \leq p_{\max}^{p_{\max}}$ and $r = p_{\max}$, which yields an algorithm with run time $p_{\max}^{p_{\max}^2} \cdot (\log n  + \log m)$; we therefore improve both the algorithm by Mnich and Wiese~\cite{MnichWiese2015} as well as the algorithm by Knop and Kouteck{\'y}~\cite{KnopKoutecky2017}.

\paragraph{Lobbying in multiple referenda.}
Bredereck et al.~\cite{BredereckEtAl2014c} study the computational social choice problem of lobbying in multiple referenda, and show a Lenstra-based fixed-parameter algorithm for the parameter $m$=``number of choices''.
The number of choices induces the parameter $\ell$=``number of ballots'', which is clearly bounded by $\ell \leq 2^{m}$.
This leads to an ILP formulation with $2^{\Oh(m)}$ variables, and thus to a run time double-exponential in $m$.

We point out that the proof of \cite[Theorem 9]{BredereckEtAl2014c} contains a combinatorial pre-$n$-fold IP, with one set of constraints indexed over $i$, and another set of constraints indexed by $j$, which are simply sums of variables over all $i$.
Since $i \in [m]$ and $j \in [\ell]$, we have the parameters $\Delta = 1$, $t=2^m$, $r=m$.
Therefore, we improve their algorithm to only single-exponential in $m$, namely $m^{\Oh(m^2)} \log n$.

\paragraph{Weighted set multicover (WSM) in graph algorithms.}
In Subsect.~\ref{subsec:wsm}, we give a combinatorial $n$-fold IP formulation for WSM.
While WSM was studied in the context of computational social choice by Bredereck et al.~\cite{BredereckEtAl2015}, it appeared implicitly several times in algorithms for restricted classes of graphs, namely graphs of bounded vertex cover number and neighborhood diversity.
We briefly mention some of these results which we improve here.

Fiala et al.~\cite{FialaGK:2011} show that \textsc{Equitable Coloring} and $L(p,1)$-\textsc{Coloring} parameterized by vertex cover are fixed-parameter tractable.
In the proof of~\cite[Theorem 4]{FialaGK:2011} and similarly of~\cite[Theorem~9]{FialaGK:2011} they construct a combinatorial pre-$n$-fold IP.
Lampis~\cite{Lampis2012} introduced the neighborhood diversity parameter and used combinatorial pre-$n$-folds to show that \textsc{Graph Coloring} and \textsc{Hamiltonian Cycle} are \FPT with this parameterization.
Ganian~\cite{Ganian:2012} later showed in a similar fashion that \textsc{Vertex Disjoint Paths} and \textsc{Precoloring Extension} are also fixed-parameter tractable parameterized by neighborhood diversity.
Similarly, Fiala et al.~\cite{FialaEtAl2017} shows that the \textsc{Uniform Channel Assignment} problem is (triple-exponential) fixed-parameter tractable parameterized by neighborhood diversity and the largest edge weight.
Ganian and Obdr{\v{z}{\'a}lek~\cite{GanianOrdyniak2016} and later Knop et al.~\cite{KnopEtAl2017c} study extensions of the MSO logic, and provide fixed-parameter algorithms for their model checking on graphs of bounded neighborhood diversity.
These algorithms again use combinatorial pre-$n$-folds under the hood.

We can speed up all aforementioned algorithms by applying Corollary~\ref{cor:combinatorial_prenfold}.





\section{Discussion}
\label{sec:discussion}
We established new and fast fixed-parameter algorithms for a class of $n$-fold IPs, which led to the first single-exponential time algorithms for several well-studied problems.
Many intriguing questions arise; e.g., is \textsc{Huge $n$-fold IP} fixed-parameter tractable for parameter $(r,s,t,\Delta)$?
One sees that optimality certification is fixed-parameter tractable using ideas similar to Onn~\cite{Onn2014}; yet, one possibly needs exponentially (in the input size) many augmenting steps.

For most of our applications, complexity lower bounds are not known to us.
Our algorithms yield complexity upper bounds of $k^{\Oh(k^2)}$ on the dependence on parameter $k$ for various problems, such as {\sc Closest String}, {\sc Weighted Set Multicover}, Score-{\sc Swap Bribery} or even {\sc Makespan Minimization}~\cite{KnopKoutecky2017}.
Is this just a common feature of our algorithm, or are there hidden connections between some of these problems?
And what are their actual complexities?
All we know so far is a trivial ETH-based $2^{o(k)}$ lower bound for {\sc Closest String} based on its reduction from {\sc Satisfiability}~\cite{FrancesLitman1997}.

\bibliographystyle{alpha}      
\bibliography{nfolds}   

\newcommand{\etalchar}[1]{$^{#1}$}
\begin{thebibliography}{KKMT17}

\bibitem[AFRS14]{AmirEtAl2014}
Amihood Amir, Jessica Ficler, Liam Roditty, and Oren~Sar Shalom.
\newblock On the efficiency of the {H}amming {C}-centerstring problems.
\newblock In {\em Proc. CPM 2014}, volume 8486 of {\em Lecture Notes Comput.
  Sci.}, pages 1--10, 2014.

\bibitem[AGST06]{AvilaEtAl2006}
Liliana~F{\'e}lix Avila, Alina Garc{\i}a, Mar{\i}a~Jos{\'e} Serna, and
  Dimitrios~M Thilikos.
\newblock A list of parameterized problems in bioinformatics.
\newblock Technical Report LSI-06-24-R, Technical University of Catalonia,
  2006.

\bibitem[ALN{\etalchar{+}}11]{AmirEtAl2011}
Amihood Amir, Gad~M. Landau, Joong~Chae Na, Heejin Park, Kunsoo Park, and
  Jeong~Seop Sim.
\newblock Efficient algorithms for consensus string problems minimizing both
  distance sum and radius.
\newblock {\em Theoret. Comput. Sci.}, 412(39):5239--5246, 2011.

\bibitem[BCF{\etalchar{+}}14]{BredereckEtAl2014}
Robert Bredereck, Jiehua Chen, Piotr Faliszewski, Jiong Guo, Rolf Niedermeier,
  and Gerhard~J. Woeginger.
\newblock Parameterized algorithmics for computational social choice: Nine
  research challenges.
\newblock {\em Tsinghua Sci. Tech.}, 19(4):358--373, 2014.

\bibitem[BCH{\etalchar{+}}14]{BredereckEtAl2014c}
Robert Bredereck, Jiehua Chen, Sepp Hartung, Stefan Kratsch, Rolf Niedermeier,
  Ondrey Such{\'y}, and Gerhard~J. Woeginger.
\newblock A multivariate complexity analysis of lobbying in multiple referenda.
\newblock {\em J. Artificial Intelligence Res.}, 50:409--446, 2014.

\bibitem[BFN{\etalchar{+}}15]{BredereckEtAl2015}
Robert Bredereck, Piotr Faliszewski, Rolf Niedermeier, Piotr Skowron, and
  Nimrod Talmon.
\newblock Elections with few candidates: Prices, weights, and covering
  problems.
\newblock In {\em Proc. ADT 2015}, volume 9346 of {\em Lecture Notes Comput.
  Sci.}, pages 414--431, 2015.

\bibitem[BHKN14]{BulteauEtAl2014}
Laurent Bulteau, Falk H{\"u}ffner, Christian Komusiewicz, and Rolf Niedermeier.
\newblock Multivariate algorithmics for $\mathsf{NP}$-hard string problems.
\newblock {\em Bulletin of the EATCS}, 114, 2014.

\bibitem[BM11]{BoucherMa2011}
Christina Boucher and Bin Ma.
\newblock Closest string with outliers.
\newblock {\em BMC Bioinformatics}, 12(S-1):S55, 2011.

\bibitem[BW10]{BoucherWilkie2010}
Christina Boucher and Kathleen Wilkie.
\newblock Why large closest string instances are easy to solve in practice.
\newblock In {\em Proc. SPIRE 2010}, volume 6393 of {\em Lecture Notes Comput.
  Sci.}, pages 106--117, 2010.

\bibitem[CFK{\etalchar{+}}15]{CyganEtAl2015}
Marek Cygan, Fedor~V. Fomin, {\L}ukasz Kowalik, Daniel Lokshtanov, D{\'a}niel
  Marx, Marcin Pilipczuk, Michal Pilipczuk, and Saket Saurabh.
\newblock {\em Parameterized algorithms}.
\newblock Springer, 2015.

\bibitem[DHK13]{DeLoeraEtAl2013}
Jesus~A. {De Loera}, Raymond Hemmecke, and Matthias K{\"o}ppe.
\newblock {\em Algebraic and Geometric Ideas in the Theory of Discrete
  Optimization}, volume~14 of {\em {MOS-SIAM} Series on Optimization}.
\newblock {SIAM}, 2013.

\bibitem[DPV11]{DadushEtAl2011}
Daniel Dadush, Chris Peikert, and Santosh Vempala.
\newblock Enumerative lattice algorithms in any norm via {M}-ellipsoid
  coverings.
\newblock In {\em Proc. FOCS 2011}, pages 580--589. 2011.

\bibitem[DS12]{DornSchlotter2012}
Britta Dorn and Ildik\'o Schlotter.
\newblock Multivariate complexity analysis of swap bribery.
\newblock {\em Algorithmica}, 64(1):126--151, 2012.

\bibitem[EFS09]{ElkindEtAl2009}
Edith Elkind, Piotr Faliszewski, and Arkadii Slinko.
\newblock Swap bribery.
\newblock In {\em Proc. SAGT 2009}, volume 5814 of {\em Lecture Notes Comput.
  Sci.}, pages 299--310, 2009.

\bibitem[ES06]{EisenbrandShmonin2006}
Friedrich Eisenbrand and Gennady Shmonin.
\newblock Carath\'eodory bounds for integer cones.
\newblock {\em Oper. Res. Lett.}, 34(5):564--568, 2006.

\bibitem[FGK11]{FialaGK:2011}
Ji{\v{r}}{\'\i} Fiala, Petr~A Golovach, and Jan Kratochv{\'\i}l.
\newblock Parameterized complexity of coloring problems: Treewidth versus
  vertex cover.
\newblock {\em Theoret. Comput. Sci.}, 412(23):2513--2523, 2011.

\bibitem[FGK{\etalchar{+}}17]{FialaEtAl2017}
Ji{\v{r}}{\'i} Fiala, Tom{\'a}{\v{s}} Gaven{\v{c}}iak, Du{\v{s}}an Knop, Martin
  Kouteck{\'y}, and Jan Kratochv{\'i}l.
\newblock Parameterized complexity of distance labeling and uniform channel
  assignment problems.
\newblock {\em Discrete Appl. Math.}, 2017.
\newblock to appear.

\bibitem[FGR12]{FellowsGR:2012}
Michael~R Fellows, Serge Gaspers, and Frances~A Rosamond.
\newblock Parameterizing by the number of numbers.
\newblock {\em Theory Comput. Syst.}, 50(4):675--693, 2012.

\bibitem[FHH06]{FaliszewskiEtAl2006}
Piotr Faliszewski, Edith Hemaspaandra, and Lane~A. Hemaspaandra.
\newblock The complexity of bribery in elections.
\newblock In {\em Proc. AAAI 2006}, pages 641--646, 2006.

\bibitem[Fis77]{Fishburn1977}
Peter~C. Fishburn.
\newblock Condorcet social choice functions.
\newblock {\em SIAM J. Appl. Math.}, 33(3):469--489, 1977.

\bibitem[FL97]{FrancesLitman1997}
Moti Frances and Ami Litman.
\newblock On covering problems of codes.
\newblock {\em Theory Comput. Syst.}, 30(2):113--119, 1997.

\bibitem[FLM{\etalchar{+}}08]{FellowsEtAl2008}
Michael~R. Fellows, Daniel Lokshtanov, Neeldhara Misra, Frances~A. Rosamond,
  and Saket Saurabh.
\newblock Graph layout problems parameterized by vertex cover.
\newblock In {\em Proc. ISAAC 2008}, volume 5369 of {\em Lecture Notes Comput.
  Sci.}, pages 294--305. Springer, Berlin, 2008.

\bibitem[Gan12]{Ganian:2012}
Robert Ganian.
\newblock Using neighborhood diversity to solve hard problems.
\newblock Technical report, 2012.
\newblock \url{https://arxiv.org/abs/1201.3091}.

\bibitem[GLO13]{GajarskyEtAl2013}
Jakub Gajarsk\'y, Michael Lampis, and Sebastian Ordyniak.
\newblock Parameterized algorithms for modular-width.
\newblock In {\em Proc. IPEC 2013}, volume 8246 of {\em Lecture Notes Comput.
  Sci.}, pages 163--176. 2013.

\bibitem[GNR03]{GrammEtAl2003}
Jens Gramm, Rolf Niedermeier, and Peter Rossmanith.
\newblock Fixed-parameter algorithms for closest string and related problems.
\newblock {\em Algorithmica}, 37(1):25--42, 2003.

\bibitem[GO16]{GanianOrdyniak2016}
Robert Ganian and Sebastian Ordyniak.
\newblock The complexity landscape of decompositional parameters for {ILP}.
\newblock In {\em Proc. AAAI 2016}, pages 710--716, 2016.

\bibitem[GOR17]{GanianEtAl2017}
Robert Ganian, Sebastian Ordyniak, and M.~S. Ramanujan.
\newblock Going beyond primal treewidth for ({M}){ILP}.
\newblock In {\em Proc. AAAI 2017}, pages 815--821, 2017.

\bibitem[GR14]{GoemansRothvoss2014}
Michel~X. Goemans and Thomas Rothvo\ss.
\newblock Polynomiality for bin packing with a constant number of item types.
\newblock In {\em Proc. SODA 2014}, pages 830--839, 2014.

\bibitem[HKW14]{HemmeckeEtAl2014}
Raymond Hemmecke, Matthias K\"oppe, and Robert Weismantel.
\newblock Graver basis and proximity techniques for block-structured separable
  convex integer minimization problems.
\newblock {\em Math. Program.}, 145(1-2, Ser. A):1--18, 2014.

\bibitem[HOR13]{HemmeckeEtAl2013}
Raymond Hemmecke, Shmuel Onn, and Lyubov Romanchuk.
\newblock {$n$}-fold integer programming in cubic time.
\newblock {\em Math. Program.}, 137(1-2, Ser. A):325--341, 2013.

\bibitem[HOW11]{HemmeckeEtAl2011}
Raymond Hemmecke, Shmuel Onn, and Robert Weismantel.
\newblock A polynomial oracle-time algorithm for convex integer minimization.
\newblock {\em Math. Program.}, 126(1, Ser. A):97--117, 2011.

\bibitem[HR15]{HermelinRozenberg2015}
Danny Hermelin and Liat Rozenberg.
\newblock Parameterized complexity analysis for the closest string with
  wildcards problem.
\newblock {\em Theoret. Comput. Sci.}, 600:11--18, 2015.

\bibitem[HS90]{HochbaumS:90}
Dorit~S. Hochbaum and J.~George Shanthikumar.
\newblock Convex separable optimization is not much harder than linear
  optimization.
\newblock {\em J. {ACM}}, 37(4):843--862, 1990.

\bibitem[JK15]{JansenKratsch2015}
Bart M.~P. Jansen and Stefan Kratsch.
\newblock A structural approach to kernels for {ILP}s: Treewidth and total
  unimodularity.
\newblock In {\em Proc. ESA 2015}, volume 9294 of {\em Lecture Notes Comput.
  Sci.}, pages 779--791, 2015.

\bibitem[JK17]{JansenK2017}
Klaus Jansen and Kim{-}Manuel Klein.
\newblock About the structure of the integer cone and its application to bin
  packing.
\newblock In {\em Proc. SODA 2017}, pages 1571--1581, 2017.

\bibitem[Kan87]{Kannan1987}
Ravi Kannan.
\newblock Minkowski's convex body theorem and integer programming.
\newblock {\em Math. Oper. Res.}, 12(3):415--440, 1987.

\bibitem[KK17]{KnopKoutecky2017}
Du{\v{s}}an Knop and Martin Kouteck{\'y}.
\newblock Scheduling meets $n$-fold integer programming.
\newblock {\em J. Sched.}, 2017.
\newblock to appear.

\bibitem[KKM17a]{KnopEtAl2017b}
Du{\v s}an Knop, Martin Kouteck{\'y}, and Matthias Mnich.
\newblock Combinatorial $n$-fold integer programming and applications.
\newblock In {\em Proc. ESA 2017}, volume~87 of {\em Leibniz Int. Proc.
  Informatics}, pages 54:1--54:14, 2017.

\bibitem[KKM17b]{KnopEtAl2017}
Du{\v{s}}an Knop, Martin Kouteck{\'y}, and Matthias Mnich.
\newblock Voting and bribing in single-exponential time.
\newblock In {\em Proc. STACS 2017}, volume~66 of {\em Leibniz Int. Proc.
  Informatics}, pages 46:1--46:14, 2017.

\bibitem[KKMT17]{KnopEtAl2017c}
Du{\v{s}}an Knop, Martin Kouteck{\'y}, Tom{\'a}{\v{s}} Masa{\v{r}}{\'\i}k, and
  Tom{\'a}{\v{s}} Toufar.
\newblock Simplified algorithmic metatheorems beyond {MSO}: Treewidth and
  neighborhood diversity.
\newblock Technical report, 2017.
\newblock \url{https://arxiv.org/abs/1703.00544}.

\bibitem[KP00]{KhachiyanPorkolab2000}
Leonid Khachiyan and Lorant Porkolab.
\newblock Integer optimization on convex semialgebraic sets.
\newblock {\em Discrete Comput. Geom.}, 23(2):207--224, 2000.

\bibitem[Kra16]{Kratsch2016}
Stefan Kratsch.
\newblock On polynomial kernels for sparse integer linear programs.
\newblock {\em J. Comput. System Sci.}, 82(5):758--766, 2016.

\bibitem[Lam12]{Lampis2012}
Michael Lampis.
\newblock Algorithmic meta-theorems for restrictions of treewidth.
\newblock {\em Algorithmica}, 64(1):19--37, 2012.

\bibitem[{Len}83]{Lenstra1983}
Hendrik~W. {Lenstra, Jr.}
\newblock Integer programming with a fixed number of variables.
\newblock {\em Math. Oper. Res.}, 8(4):538--548, 1983.

\bibitem[Lok15]{Lokshtanov2015}
Daniel Lokshtanov.
\newblock Parameterized integer quadratic programming: Variables and
  coefficients.
\newblock Technical report, 2015.
\newblock \url{http://arxiv.org/abs/1511.00310}.

\bibitem[Mar08]{Marx2008}
D\'aniel Marx.
\newblock Closest substring problems with small distances.
\newblock {\em SIAM J. Comput.}, 38(4):1382--1410, 2008.

\bibitem[MW15]{MnichWiese2015}
Matthias Mnich and Andreas Wiese.
\newblock Scheduling and fixed-parameter tractability.
\newblock {\em Math. Program.}, 154(1-2, Ser. B):533--562, 2015.

\bibitem[Nie04]{Niedermeier2004}
Rolf Niedermeier.
\newblock Ubiquitous parameterization---invitation to fixed-parameter
  algorithms.
\newblock In {\em Proc. MFCS 2004}, volume 3153 of {\em Lecture Notes Comput.
  Sci.}, pages 84--103. 2004.

\bibitem[NS12]{NishimuraSimjour2012}
Naomi Nishimura and Narges Simjour.
\newblock Enumerating neighbour and closest strings.
\newblock In {\em Proc. IPEC 2012}, volume 7535 of {\em Lecture Notes Comput.
  Sci.}, pages 252--263. Springer, Heidelberg, 2012.

\bibitem[Onn10]{Onn2010}
Shmuel Onn.
\newblock Nonlinear discrete optimization.
\newblock {\em Zurich Lectures in Advanced Mathematics, European Mathematical
  Society}, 2010.

\bibitem[Onn14]{Onn2014}
Shmuel Onn.
\newblock Huge multiway table problems.
\newblock {\em Discrete Optim.}, 14:72--77, 2014.

\bibitem[OS15]{OnnSarrabezolles2015}
Shmuel Onn and Pauline Sarrabezolles.
\newblock Huge unimodular {$n$}-fold programs.
\newblock {\em SIAM J. Discrete Math.}, 29(4):2277--2283, 2015.

\end{thebibliography}


\end{document}